\documentclass[letterpaper,11pt]{article}
\usepackage[letterpaper,hmargin=1in,vmargin=1in,footskip=36pt]{geometry}
\usepackage{amsthm, amssymb, amsmath}
\usepackage[latin9]{inputenc}
\usepackage{txfonts}
\usepackage{cite} %
\usepackage{verbatim}
\usepackage{algorithmic}
\usepackage{algorithm}
\usepackage{paralist} %
\usepackage{tabularx,booktabs}
\newcolumntype{Y}{>{\raggedright\arraybackslash}X}
\newcolumntype{Z}{>{\centering\arraybackslash}X}
\usepackage{times}
\numberwithin{equation}{section}
\numberwithin{figure}{section}
\theoremstyle{plain}
\newtheorem{thm}{Theorem}
\numberwithin{thm}{section}  %
\theoremstyle{plain}
\newtheorem{lem}[thm]{Lemma}
\theoremstyle{plain}
\newtheorem*{lem-safety-net}{Lemma \ref{lem:safety-net}}
\theoremstyle{plain}
\newtheorem{prop}[thm]{Proposition}
\theoremstyle{plain}

\theoremstyle{definition}
\newtheorem{defn}[thm]{Definition}

\newcommand{\npc}{{NP}}

\newcommand{\OO}{\ensuremath{\mathcal{O}}}
\newcommand{\algoname}[1]{\textsf{\textsc{#1}}}
\newcommand{\opt}{\algoname{Opt}}
\newcommand{\alg}{\algoname{A}}%

\newcommand{\Exp}[1]{\mathbb{E}\left[\,#1\,\right]}

\newcommand{\I}{\ensuremath{\mathcal{I}}} %

\newcommand{\abc}[3]{\ensuremath{\textup{#1}|\,#2\,|#3}}
\newcommand{\Part}[1]{\ensuremath{P_{\!#1}}}%
\newcommand{\first}[1]{\ensuremath{\textnormal{first}(#1)}}%
\newcommand{\Ir}{\ensuremath{\mathrm{Ir}}}
\newcommand{\Rel}{\ensuremath{\mathrm{Rel}}}
\newcommand{\IS}{\ensuremath{\mathcal{S}}}

\newcommand{\C}{\ensuremath{\mathcal{C}}}

\global\long\def\N{\mu}

\global\long\def\R{E}

\newcommand{\Q}{\ensuremath{\mathbb{Q}}}
\newcommand{\eps}{\varepsilon}
\renewcommand{\epsilon}{\eps}

\newcommand{\ee}{\left(1+\eps\right)}
\newcommand{\e}{1+\eps}
\newcommand{\eo}{1+O(\eps)}

\usepackage{xcolor}
\title{A New Approach to Online Scheduling:\\ Approximating the Optimal Competitive Ratio}
\author{ Elisabeth G\"unther\thanks{Department of Mathematics, Technische Universit{\"a}t
       Berlin, Germany.  Email: \{\texttt{eguenth,maurer,nmegow}\}\texttt{@math.tu-berlin.de}.}
     \hspace{1pt}\thanks{Supported by the DFG Research center \textsc{Matheon}
        {\em Mathematics for key technologies} in Berlin.}%
      \and %
      Olaf Maurer\footnotemark[1]\,\hspace{2pt}\footnotemark[2] \and %
      Nicole Megow\footnotemark[1]\,\hspace{2pt}\thanks{Supported by the German Science Foundation (DFG) under contract  ME 3825/1.}\and %
      Andreas Wiese\thanks{Department of Computer and System Sciences, Sapienza
        University of Rome, Italy. Email:
        \texttt{wiese@dis.uniroma1.it}.}
      \hspace{1.5pt}\thanks{Supported by the DFG Focus Program 1307 and by a fellowship within the Postdoc-Programme of the 
German Academic Exchange Service (DAAD).}}

\begin{document}
\date{}
\maketitle
\thispagestyle{empty}

\begin{abstract}
  We propose a new approach to competitive analysis in online scheduling by introducing the
  novel concept of competitive-ratio approximation schemes. Such a scheme
  algorithmically constructs an online algorithm with a competitive ratio arbitrarily close to the best possible competitive ratio
  for \emph{any} online algorithm.  We study the problem of scheduling
  jobs online to minimize the weighted sum of completion times on
  parallel, related, and unrelated machines, and we derive both
  deterministic and randomized algorithms which are almost best
  possible among all online algorithms of the respective settings. We also generalize our techniques to arbitrary monomial cost functions and apply them to the makespan objective. Our
  method relies on an abstract characterization of online algorithms
  combined with various simplifications and
  transformations. %
We also contribute algorithmic means to compute the actual value of 
the best possible competitive ratio up to an arbitrary accuracy. 
This strongly contrasts (nearly) all previous manually obtained competitiveness
results and, most importantly, it reduces the search
for the optimal competitive ratio to a question that a computer can
answer. We believe that our concept can also be applied to many other problems
and yields a new perspective on online algorithms in general.
We believe that our concept can also be applied to many other problems
and yields a new perspective on online algorithms in general.
\end{abstract}

\textbf{Keywords}: online scheduling, competitive analysis, jobs arrive over time, min-sum objective, makespan

\newpage
\setcounter{page}{1}
\section{Introduction}

Competitive analysis~\cite{sleatorT85,karlin88} is the most popular
method for studying the performance of online algorithms. It provides
an effective framework to analyze and classify algorithms based on
their worst-case behavior compared to an optimal offline algorithm
over an infinite set of input instances.  For some problem types,
e.g., online paging, competitive analysis may not be adequate to
evaluate the performance of algorithms, but for a vast majority of
online problems it is practical, natural, and yields meaningful
results. A classical such problem is online scheduling to minimize the weighted average completion time. It
has received a lot of attention in the past two decades. For different
machine environments, a long sequence of papers emerged introducing
new techniques and algorithms, improving upper and lower bounds on the
competitive ratio of particular algorithms as well as on the best
possible competitive ratio that any online algorithm can achieve.
Still, unsatisfactory gaps remain. As for most online problems, a
provably optimal online algorithm, w.r.t.\ competitive analysis,
among {\em all} online algorithms is only known for very special
cases.

In this work we close these gaps and present nearly optimal online
scheduling algorithms. We provide {\em competitive-ratio approximation schemes} that compute algorithms
with a competitive ratio that is at most a factor~$1+\eps$ larger than
the optimal ratio for any~$\eps>0$. To that end, we introduce a new
way of designing online algorithms. Apart from structuring and
simplifying input instances, we find an abstract description of online scheduling
algorithms, which allows us to reduce the infinite-size set of all online
algorithms to a relevant set of finite size. This is the key for
eventually allowing an enumeration scheme that finds an online
algorithm with a competitive ratio arbitrarily close to the optimal
one. 
Besides improving on previous algorithms, our method also
provides an algorithm to compute the competitive ratio of the designed
algorithm, and even the best possible competitive ratio, up to any
desired accuracy. This is clearly in
strong contrast to all previously given (lower) bounds that stem from
{\em manually} designed input instances. 
We are aware of only very few online problems for which a competitive ratio, or even the optimal competitive ratio,
   are known to be {\em computable} by some algorithm (for a not inherently finite problem).
   Our result is surprising, as there are typically no means of enumerating all possible input
   instances and all possible online algorithms. Even for only one given
   algorithm, usually one cannot compute its competitive ratio 
   simply due to difficulties like the halting problem. We overcome these issues and pave the way for computer-assisted design
  of~online~algorithms.

We believe that our concept of abstraction for online algorithms can be
applied successfully to other %
problems. 
We show this for other scheduling problems with jobs arriving online over time. We hope that our new approach to competitive
analysis contributes to a better understanding of online algorithms
and may lead to a new line of~research~in~online~optimization.

\subsection{Problem Definition and Previous Results}

{\bf Competitive analysis.} Given a minimization problem, a
deterministic online algorithm \alg\ is called~{\em
  $\rho$-competitive} if, for any problem instance~$\I$, it achieves a
solution of value~$\alg(\I) \leq \rho \cdot \opt(\I)$,
where~$\opt(\I)$ denotes the value of an optimal offline solution for
the same instance~$\I$. A randomized online algorithm is
called~$\rho$-competitive, if it achieves in expectation a solution of
value~$\Exp{\alg(\I)} \leq \rho \cdot \opt(\I)$ for any instance~$\I$.
The \emph{competitive ratio}~$\rho_{\alg}$ of~$\alg$ is the infimum
over all~$\rho$ such that~$\alg$ is~$\rho$-competitive. The minimum
competitive ratio~$\rho^*$ achievable by any online algorithm is
called {\em optimal}. 
Note that there are no requirements on the
  computational complexity of competitive algorithms. Indeed, the
    competitive ratio measures the best possible performance under the
    lack of information given unbounded computational resources.

We define a competitive-ratio approximation scheme as a procedure that computes
a nearly optimal online algorithm and at the same time provides a
nearly exact estimate of the optimal competitive ratio.
\begin{defn}
  A {\em competitive-ratio approximation scheme} computes for a given~$\eps>0$
  an online algorithm~$\alg$ with a competitive ratio~$\rho_{\alg}
  \leq (1+\eps)\rho^*$. Moreover, it determines a value~$\rho'$ such
  that $\rho'\leq \rho^* \leq (1+\eps) \rho'$.
\end{defn}

\noindent {\bf Online scheduling.} A scheduling instance consists of a
fixed set of~$m$ machines and a set of jobs~$J$, where each job~$j\in
J$ has processing time~$p_j\in \Q^+$, weight~$w_j\in \Q^+$, and %
release date~$r_j\in \Q^+$. %
The jobs arrive online over time, i.e., each job becomes known to
the scheduling algorithm only at its release date. We consider three
different machine environments: identical parallel machines~(denoted
by~P), related machines~(Q) where each machine $i$ has associated a
speed~$s_i$ and processing a job~$j$ on machine~$i$ takes $p_j/s_i$
time, and unrelated machines~(R) where the processing time of a job
$j$ on each machine $i$ is explicitly given as a value~$p_{ij}$. The main problem considered in this paper is to schedule the jobs on the given set of machines so as
to minimize $\sum_{j\in J}w_jC_j$, where~$C_j$ denotes the completion
time of job~$j$. We consider the problem with and without preemption.
Using standard scheduling notation~\cite{grahamLLR79}, we denote
the non-preemptive (preemptive) problems that we consider in this
paper by \abc{Pm}{r_j, (pmtn)}{\sum w_jC_j}, \abc{Qm}{r_j,
  (pmtn)}{\sum w_jC_j}, and \abc{Rm}{r_j,pmtn}{\sum w_jC_j}.
We also briefly consider more general min-sum objectives~$\sum_{j\in J}w_jf(C_j)$, where~$f$ is an arbitrary monomial function~$f(x)=k\cdot x^{\alpha}$, with constant~$\alpha\geq 1, k >0$, and the classical makespan~\mbox{$C_{\max}:=\max_{j\in J}C_j$}.

\medskip
\noindent {\bf Previous results.}
The offline variants of nearly all
problems under consideration are
\npc-hard~\cite{GareyJohnson79,labetoulle84,lenstra77}, but in most
cases polynomial-time approximation schemes have been
developed~\cite{afrati99,LenstraShmoysTardos90,HochbaumShmoys87,HochbaumShmoys88}.

The corresponding online settings have been a highly active field of research in the
past fifteen years. A whole sequence of papers appeared introducing
new algorithms, new relaxations and analytical techniques that
decreased the gaps between lower and upper bounds on the optimal
competitive
ratio~\cite{goemans02,schulzS02,hallSSW97,sitters10-ipco,sitters10-orl,andersonP04,
  chekuriMNS01,hoogeveenV96,correaW09,goemans97,megowS04,megow06-diss,chungNS10,
  schulzS02-unrelated,liuL09,luSS03,stougieV02,phillipsST98,chakrabartiPSSSW96,
  seiden00,epsteinvS01,chenV97}.
Interestingly, despite the considerable effort, optimal competitive
ratios are known only for~\abc{1}{r_j, pmtn}{\sum
  C_j}~\cite{schrage68} and for non-preemptive single-machine
scheduling~\cite{andersonP04,stougieV02,hoogeveenV96,chekuriMNS01}.
In all other scheduling settings remain unsatisfactory, 
even~quite~significant~gaps. 
See Appendix~\ref{app:related_work} 
for a detailed description of the 
state of the art for each individual problem.

Very recently, our new concept of competitive-ratio approximation
schemes was applied 
also to job shop scheduling \abc{Jm}{r_j,op\leq \mu}{C_{\max}} and
non-preemptive scheduling on unrelated machines~\abc{Rm}{rj}{C_{\max}}~\cite{kurpiszMG12}.

To the best of our knowledge, there are only very few problems
  in online optimization for which an optimal competitive ratio can be
  determined, bounded, or approximated by computational means. Lund and Reinhold~\cite{LundReingold1994}
present a framework for upper-bounding the optimal competitive ratio
of randomized algorithms 
by a linear program. For certain cases, e.g., the 2-server problem in a space of
three points, this yields a provably optimal competitive
ratio. Ebenlendr et al. \cite{EbenlendrJaworSgall2009,EbenlendrS11}
study various online and semi-online variants of scheduling preemptive
jobs on uniformly related machines to minimize the makespan.
In contrast to our model, they assume the jobs to be given one by one (rather than over time). They
prove that the optimal competitive ratio can be computed by a linear
program for any given set of speeds. In terms of approximating the
best possible performance guarantee, the work by Augustine, Irani, and
Swamy~\cite{AugustineIraniSwamy2008} is closest to ours. They show how
to compute a nearly optimal power-down strategy for a processor with a
finite number of power states.

\subsection{New Results and Methodology}
\vspace{-0.5ex}
In this paper, we introduce the concept of competitive-ratio approximation schemes
and present such schemes for various scheduling problems with jobs arriving online over time. 
We present our technique focussing on the problems~\abc{Pm}{r_j, (pmtn)}{\sum w_jC_j}, \abc{Qm}{r_j,
(pmtn)}{\sum w_jC_j}~(assuming a constant range of machine speeds 
without preemption), and \abc{Rm}{r_j,pmtn}{\sum w_jC_j}, and we 
comment on how it applies to other cost functions such as the 
makespan, $C_{\max}$, and~$\sum_{j\in J}w_jf(C_j)$, where~$f$ is an 
arbitrary monomial function with fixed exponent. %
For any~$\eps>0$, we show that the competitive ratios of our new algorithms
are by at most a factor~$\e$ larger than the respective
optimal competitive ratios.
We obtain such nearly optimal online algorithms for the deterministic
as well as the randomized setting, for any number of machines
$m$.Moreover, we give an algorithm
which estimates the optimal competitive ratio for these problems to any desired
accuracy. Thus, we reduce algorithmically the performance gaps for all considered problems 
to an arbitrarily small value. These results reduce the long-time ongoing search for the best possible competitive ratio for the
considered problems to a question that can be answered by~a~finite~algorithm.

To achieve our results, we introduce a new and unusual way of designing
online scheduling algorithms. We present an abstraction in which online
algorithms are formalized as \emph{algorithm maps}. Such a map
receives as input a set of unfinished jobs %
together with the schedule computed so far. Based on this information,
it returns a schedule for the next time instant. This view captures
exactly how online algorithms operate under limited information. The
total number of algorithm maps is unbounded. However, we show that
there is a finite subset which approximates the entire set.  More
precisely, for any algorithm map there is a map in our subset whose
competitive ratio is at most by a factor~$\e$ larger. To achieve this
reduction, we first apply several standard techniques, such as
geometric rounding, time-stretch, and weight-shift,
to transform and simplify the input problem without increasing the
objective value too much; see, e.g.,~\cite{afrati99}. The key,
however, is the insight that it suffices for an online algorithm to
base its decisions on the currently unfinished jobs and a very {\em
  limited part} of the so far computed schedule---rather than the
entire history. This allows for an enumeration of all relevant
algorithm maps (see also \cite{ManasseMcGeochSleator1988} for an enumeration routine
for online algorithms for a fixed task system with finitely many states). 
For randomized algorithms we even show that we can
restrict to instances with only constantly many jobs. As all our
structural insights also apply to offline algorithms for the same
problems, %
they might turn out to be useful for other~settings~as~well.

Our algorithmic scheme contributes more than %
an improved competitive ratio. It also outputs~(up to a
factor~$1+\eps$) the exact value of the competitive ratio of the
derived algorithm, which implies a $(1+\eps)$-estimate for the optimal
competitive ratio.
 This contrasts strongly all earlier results
where~(matching) upper and lower bounds on the competitive ratio of a
particular and of all online algorithm
 had to be derived {\em manually}, instead of
executing an algorithm using, e.g., a computer.
In general, there are no computational means to determine the
competitive ratio of an algorithm---even when it is a constant. %
It is simply not 
possible to enumerate all possible input instances.
Even more, there are no general means of enumerating all possible online algorithms to
determine the optimal competitive ratio.
However, for the scheduling problems studied in this paper our
extensive simplification of input instances and our abstract view on
online algorithms allow us to overcome these obstacles, losing only a
factor of $1+\epsilon$ in the objective.

Although the enumeration scheme for identifying the (nearly)
  optimal online algorithm heavily exploits unbounded computational
resources, the resulting algorithm itself has polynomial %
running time. As a consequence, there are efficient online algorithms for the considered problems with
almost optimal competitive ratios. Hence, the granted additional, even unbounded, computational power of online algorithms does not yield any significant benefit here.


\medskip
\noindent {\bf Outline of the paper.} In Section~\ref{sec:simplifications} we introduce
several general transformations and observations that simplify the structural complexity of online scheduling in the setting of~\abc{Pm}{r_j, pmtn}{\sum w_jC_j}. Based on this, we present our
abstraction of online algorithms and develop a competitive-ratio approximation
scheme in Section~\ref{sec:Abstraction-online-algos}. Next, we sketch in
Section~\ref{sec:Extensions} how to extend these techniques %
to the non-preemptive setting and more general machine environments such that 
the approximation scheme~(Sec.~\ref{sec:Abstraction-online-algos}) remains applicable. 
In Section~\ref{sec:randomized}, we present competitive-ratio approximation schemes for the 
randomized setting. Finally, in Section~\ref{sec:other-objectives}, we extend our results to other 
objective functions.
\vspace{-1ex}

\section{General Simplifications and Techniques}
\label{sec:simplifications}
\vspace{-1ex}

In this section, we discuss several transformations that simplify the
input and reduce the structural complexity of online schedules
for~\abc{Pm}{r_j, pmtn}{\sum w_jC_j}.  Later, we outline how to adapt
these for more general settings.
Our construction combines several transformation techniques known for offline PTASs~(see~\cite{afrati99} and the references therein) and a new technique to subdivide an instance online into parts which can be handled separately.
We will use the terminology that {\em at~$\eo$ loss we can restrict} to instances or schedules with certain properties.
This means that we lose at most a factor~$\eo$, as $\eps\rightarrow 0$, by
limiting our attention to those. We bound several relevant
parameters by constants. If not stated differently, any mentioned constant
depends only on~$\eps$ and~$m$.

\begin{lem}[\!\!\cite{afrati99}]
\label{lem:first-rounding}At $1+O(\eps)$ loss we can restrict
to instances where all processing times, release dates, and weights are powers of $1+\eps$, no job is released before time $t=1$, and $r_{j}\ge\eps\cdot p_{j}$ for all jobs $j$.
\end{lem}

This standard \emph{geometric rounding} procedure used in the lemma above
allows us to see intervals of the form~$I_{x}:=[R_{x},R_{x+1})$
with~$R_{x}:=(1+\eps)^{x}$ as atomic entities.
An online algorithm can define the corresponding schedule at the
beginning of an interval since no further jobs are released until the
next interval. Moreover, we assume at $\e$ loss that all jobs that finish within~$I_{x}$ have
completion time~$R_{x+1}$.

\paragraph{Simplification within intervals.}%
\label{par:simple-interval-input}

Our goal is to reduce the number of situations that can arise at the
beginning of an interval. To this end, we partition the set of jobs
released at time~$R_x$ into the set of \emph{large}
jobs~$L_x$, %
with processing times at least~$\eps^3R_x$, and the set of
\emph{small} jobs~$S_{\!x}$ with all remaining jobs.  Running {\em
  Smith's Rule}~\cite{smith56} on small jobs allows us to group very
small jobs to job \emph{packs}, which we treat as single jobs.
Together with Lemma~\ref{lem:first-rounding} we obtain bounds on the
lengths of jobs of each release date.

\begin{lem}\label{lem:constant-num-proc-times}%
  At~$\eo$ loss, we can assume that for each interval $I_x$ there are
  lower and upper bounds for the lengths of the jobs $S_x \cup L_x$
  that are within a constant factor of $R_x$ and the constants are
  independent of~$x$. Also, the number of distinct processing times of
  jobs in each interval is upper-bounded by a constant.
\end{lem}
We look for jobs in $S_{\!x}$ and $L_x$ which can be
excluded from processing within~$I_x$ at a loss of not more
than~$\eo$.  This allows us to bound the number of released jobs per
interval.

\begin{lem}\label{lem:constant-num-jobs-released}%
  At $\eo$ loss, we can restrict to instances where for each $x$, the
  number of jobs released at time $R_x$ is bounded by a
  constant~$\Delta$.
\end{lem}

To prove the above lemmas, we use the technique of~\emph{time-stretching},
see~\cite{afrati99}. %
In an online interpretation of this method, we shift the work assigned
to any interval~$I_x$ to the interval~$I_{x+1}$. This can be done at a
loss of $1+\eps$ and we obtain free space of size~$\eps\cdot I_{x'-1}$
in each interval $I_{x'}$. Again using time-stretching, we can show that no
job needs to be completed later than constantly many intervals after its
release interval.
\begin{lem}
  \label{lem:safety-net}There is a constant $s$ such that at $\eo$
  loss we can restrict to schedules such that for each interval
  $I_{x}$ there is a subinterval of $I_{x+s-1}$ 
  which is large enough to process all jobs released at~$R_x$ and during which only those jobs are executed. 
  We call this subinterval the \emph{safety net} of
  interval $I_{x}$. We can assume that each job released at~$R_x$ finishes before time~$R_{x+s}$.
\end{lem}

We can also simplify the complexity of the computed schedules by
limiting the way jobs are preempted.
We say that two large jobs are of the same \emph{type} if they have
the same processing time and the same release date. A job is
\emph{partially processed} if it has been processed, but not yet
completed.
\begin{lem}
  \label{lem:large-job-atoms}There is a constant $\N>0$ such that at
  $1+O(\epsilon)$ loss we can restrict to schedules such that
  \begin{compactitem}
  \item at the end of each interval, there are at most $m$ large jobs
    of each type which are partially processed and each of them is
    processed to an extent which is a multiple of $p_{j}\cdot\N$ and
  \item each small job finishes without preemption in the same
    interval where it started.
  \end{compactitem}
\end{lem}

\paragraph{\label{sub:periods-of-schedule}Irrelevant history.}

The schedule that an online algorithm computes for an interval may
depend on the set of currently unfinished jobs and possibly the entire
schedule used so far.  In the remainder of this section we show why we
can assume that an online algorithm only takes a finite amount of
history into account in its decision making, namely, the jobs with
relatively large weight released in the last constantly many
intervals.

Our strategy is to partition the time horizon into \emph{periods}.
For each integer~$k\ge0$, we define a period~$Q_{k}$ which consists of
the~$s$ consecutive intervals~$I_{k\cdot s},...,I_{(k+1)\cdot s-1}$.
For ease of notation, we will treat a period~$Q$ as the set of jobs
released in that period. For a set of jobs~$J$ we denote
by~$rw(J):=\sum_{j\in J}r_{j}w_{j}$ their \emph{release weight}. Note
that $rw(J)$ forms a lower bound on the quantity that these jobs must
contribute to the objective in \emph{any} schedule%
. Due to Lemma~\ref{lem:safety-net}, we also obtain an upper bound
of~$\ee^{s}\cdot rw(J)$ for the latter quantity.

\begin{lem}
\label{lem:dominated-periods}%
Let~$Q_{k},...,Q_{k+p}$ be consecutive periods such that period~$Q_{k+p}$ is the
first of this series with
\(rw(Q_{k+p})\le\frac{\epsilon}{\ee^{s}}\cdot\sum_{i=0}^{p-1}rw(Q_{k+i}).\)
Then at~$\e$ loss we can move all jobs in~$Q_{k+p}$ to their safety
nets.
\end{lem}

The above observation defines a natural partition of a given instance~$\I$ into
\emph{parts} by the insignificant periods.
Formally, let $a_{1},...,a_{\ell}$ be all ordered indices such that~$Q_{a_{i}}$ is insignificant compared to the preceding periods according to Lemma~\ref{lem:dominated-periods}~($a_{0}:=0$). Let $a_{\ell+1}$ be the index of the last period.
For each~$i\in\{0,...,\ell\}$ we define a part~$\Part{i}$
consisting of all periods~$Q_{a_{i}+1},...,Q_{a_{i+1}}$.
Again, identify with~$\Part{i}$ all jobs released in this part.
We treat now each part~$\Part{i}$ as a separate instance that we present
to a given online algorithm. For the final output, we concatenate
the computed schedules for the different parts.
It then suffices to bound~$\alg(\Part{i})/\opt(\Part{i})$
for each part~$\Part{i}$ since~$\alg(\I)/\opt(\I)\le\max_{i}\{\alg(\Part{i})/\opt(\Part{i})\}
\cdot(1+O(\epsilon)).$

\begin{lem}
  \label{lem:split-the-instance-pmtn}%
  At~$1+O(\epsilon)$ loss we can restrict to instances which consist of
  only one part.
\end{lem}

Each but the last period of one part fulfills the opposite condition of the one
from Lemma~\ref{lem:dominated-periods}. This implies exponential growth for the
series of partial sums of release weights (albeit with a small growth factor).
From this observation, we get:

\begin{lem}
\label{lem:K-periods}
There is a constant $K$ such that the following holds:
Let~$Q_{1},Q_{2},...,Q_{p}$ be consecutive periods such
that~$rw(Q_{i+1})>\frac{\epsilon}{\ee^{s}}\cdot\sum_{\ell=1}^{i}rw(Q_{\ell})$
for all~$i$. Then in any schedule
in which each job~$j$ finishes 
no later than by time~$r_{j}\cdot(1+\epsilon)^{s}$
(e.g., using the safety net)
it holds
that~$\sum_{i=1}^{p-K-1}\sum_{j\in
Q_{i}}w_{j}C_{j}\le\epsilon\cdot\sum_{i=p-K}^{p}\sum_{j\in Q_{i}}w_{j}C_{j}.$
\end{lem}

The objective value of one part is therefore dominated by the contribution
of the last $K$ periods of this part.
We will need this later to show that
at $\e$ loss we can assume that an online algorithm bases its decisions
only on a constant amount of information. Denote the corresponding
number of important intervals by~$\Gamma := Ks$.

This enables us to partition
the jobs into relevant and irrelevant jobs. Intuitively, a job is irrelevant if
it is released very early (cf.\ Lemma~\ref{lem:K-periods}) or its weight is very
small in comparison to some other job.
The subsequent lemma states that the irrelevant jobs can almost be
ignored for the objective value of a schedule.

\begin{defn}
\label{def:irrelevant_jobs}
A job~$j$ is \emph{irrelevant at time~$R_{x}$} if it was irrelevant
at time~$R_{x-1}$, or~$r_{j}<R_{x-\Gamma}$, or
it is \emph{dominated at time~$R_{x}$}. This is the case if
there is a job~$j'$,
either released at time~$R_{x}$ or already relevant at time~$R_{x-1}$ with
release date at least~$R_{x-\Gamma}$, such that%
~$w_{j}<\frac{\epsilon}{\Delta\cdot\Gamma\cdot(1+\epsilon)^{\Gamma+s}}w_{j'}$.
Otherwise, a job released until~$R_x$
is \emph{relevant at time~$R_{x}$.}
Denote the respective subsets of some job set~$J$ by~$\Rel_{x}(J)$
and~$\Ir_{x}(J)$.
\end{defn}

\begin{lem}\label{lem:irrelevant-jobs}
Consider a schedule of one part in which no job~$j$ finishes later than at
time~$r_{j}\cdot(1+\epsilon)^{s}$ (e.g., using the safety net) and let~$x$ be an interval index
in this part. Then $\sum_{j\in\Ir_{x}(J)}w_{j}C_{j}\le
O(\epsilon)\cdot\sum_{j\in\Rel_{x}(J)}w_{j}C_{j}$.
\end{lem}

The above lemma implies that at $1+O(\epsilon)$ loss we can restrict to
online algorithms which schedule the remaining part of a job in its safety net,
once it has become irrelevant.

\section{Abstraction of Online Algorithms}
\label{sec:Abstraction-online-algos}
\vspace{-1.5ex}

In this section we show how to construct a competitive-ratio approximation
scheme based on the observations of Section~\ref{sec:simplifications}.
To do so, we restrict ourselves to such simplified instances and schedules.
The key idea is to characterize the behavior of an online algorithm by
a map: For each interval, the map gets as input the schedule computed
so far and all information about the currently unfinished jobs.  Based
on this information, the map outputs how to schedule the available
jobs within this interval.

More precisely, we define the input by a \emph{configuration} and the
output by an \emph{interval-schedule}.
\begin{defn}
  An \emph{interval-schedule}~$S$ for an interval~$I_{x}$ is defined
  by
  \begin{compactitem}

  \item the index~$x$ of the interval,
  \item a set of jobs~$J(S)$ available for processing in~$I_x$ together with the
    properties~$r_{j}, p_{j}, w_{j}$ of each job~$j\in J(S)$ and its already
    finished part~$f_{j}<p_{j}$ up to~$R_{x}$,
  \item for each job~$j\in J(S)$ the information whether~$j$ is relevant
    at time~$R_{x}$, and
  \item for each job $j\in J(S)$ and each machine $i$ a value $q_{ij}$ specifying
    for how long $j$ is processed by~$S$ on machine $i$ during~$I_{x}$.
  
    \end{compactitem}
      An interval-schedule is called \emph{feasible} if there is a feasible
  schedule in which the jobs of~$J(S)$ are processed corresponding
  to the~$q_{j}$ values within the interval~$I_{x}$. Denote the
  set of feasible interval-schedules as~$\IS$.
\end{defn}
\begin{defn}
  A \emph{configuration}~$C$ for an interval~$I_{x}$ consists of
  \begin{compactitem}
  \item the index~$x$ of the interval,
  \item a set of jobs~$J(C)$ released up to time~$R_{x}$ together with the
    properties~$r_{j}, p_{j}, w_{j}, f_{j}$ of each job~$j\in J(C)$,
  \item an interval-schedule for each interval~$I_{x'}$ with~$x'<x$.
  \end{compactitem}
  The set of all configurations is denoted by~$\C$. An \emph{end-configuration} is a configuration~$C$ for an interval~$I_{x}$ such that at time~$R_{x}$,
  and not earlier, no jobs are left unprocessed.
\end{defn}

We say that an interval-schedule~$S$ is \emph{feasible} \emph{for a
configuration~$C$} if
the set of jobs in~$J(C)$ which are unfinished at time~$R_{x}$
matches the set~$J(S)$ with respect to release dates, total and remaining
processing time, weight and relevance of the jobs.

Instead of online algorithms we work from now on with~\emph{algorithm
  maps}, which are defined as functions~$f:\C\rightarrow\IS$. An
algorithm map determines a schedule~$f(\I)$ for a given scheduling
instance~$\I$ by iteratively applying~$f$ to the corresponding
configurations.  W.l.o.g.\ we consider only algorithm maps~$f$ such
that~$f(C)$ is feasible for each configuration~$C$ and~$f(\I)$ is
feasible for each instance~$I$.
Like for online algorithms, we define the competitive ratio~$\rho_{f}$
of an algorithm map~$f$ by~$\rho_{f}:=\max_{\I}f(\I)/\opt(\I)$. Due to
the following observation, algorithm maps are a natural generalization
of online algorithms.
\begin{prop}
  For each online algorithm~$\alg$ there is an algorithm
  map~$f_{\alg}$ such that when~$\alg$ is in configuration~$C\in\C$ at
  the beginning of an interval~$I_{x}$, algorithm~$\alg$ schedules the
  jobs according to~$f_{\alg}(C)$.
\end{prop}

Recall, that we restrict our attention to algorithm maps describing
online algorithms which obey the simplifications introduced in
Section~\ref{sec:simplifications}. The essence of such online
algorithms are the decisions for the relevant jobs. To this end, we
define equivalence classes for configurations and for interval-schedules.
Intuitively, two interval-schedules (configurations) are equivalent if we can
obtain one from the other by scalar multiplication with the same
value, while ignoring the irrelevant jobs.
\begin{defn}\label{def:equivalence-interval-schedule}
  Let~$S,S'$ be two feasible interval-schedules for two
  intervals~$I_{x},I_{x'}$.  Denote by~$J_{\Rel}(S)\subseteq J(S)$
  and~$J_{\Rel}(S')\subseteq J(S')$ the relevant jobs in~$J(S)$
  and~$J(S')$. Let further~$\sigma:J_{\Rel}(S)\rightarrow
  J_{\Rel}(S')$ be a bijection and $y$ an integer. The
  interval-schedules~$S,S'$ are~$(\sigma,y)$-\emph{equivalent}
  if~$r_{\sigma(j)}=r_{j}(\e)^{x'-x}, p_{\sigma(j)}=p_{j}(\e)^{x'-x},
  f_{\sigma(j)}=f_{j}(\e)^{x'-x}, q_{\sigma(j)}=q_{j}(\e)^{x'-x}$
  and~$w_{\sigma(j)}=w_{j}\ee^{y}$
  for all~$j\in J_{\Rel}(S)$.  The interval-schedules $S,S'$ are
  \emph{equivalent} (denoted by $S\sim S'$) if a map $\sigma$ and an
  integer $y$ exist such that they are $(\sigma,y)$-equivalent.
\end{defn}

\begin{defn}
  Let~$C,C'$ be two configurations for two intervals~$I_{x},I_{x'}$.
  Denote by~$J_{\Rel}(C),J_{\Rel}(C')$ the jobs which are relevant at
  times~$R_{x},R_{x'}$ in~$C,C'$, respectively. Configurations~$C,C'$
  are \emph{equivalent} (denoted by $C\sim C'$) if there is a
  bijection~$\sigma:J_{\Rel}(C)\rightarrow J_{\Rel}(C')$ and an
  integer~$y$ such that
  \begin{compactitem}
  \item $r{}_{\sigma(j)}=r_{j}(1+\epsilon)^{ x'-x },
    p{}_{\sigma(j)}=p_{j}(1+\epsilon)^{x'-x},
    f{}_{\sigma(j)}=f_{j}(1+\epsilon)^{x'-x}$
    and~$w_{\sigma(j)}=w_{j}\ee^{y}$ for all~$j\in J_{\Rel}(C)$, and
  \item the interval-schedules of~$I_{x-k}$ and~$I_{x'-k}$
    are~$(\sigma,y)$-equivalent when restricted to the jobs
    in~$J_{\Rel}(C)$ and~$J_{\Rel}(C')$ for each~$k\in\mathbb{N}$.
  \end{compactitem}
\end{defn}

A configuration~$C$ is \emph{realistic} for an algorithm map~$f$ if
there is an instance~$\I$ such that if~$f$ processes~$\I$ then at
time~$R_x$ it is in configuration~$C$.  The following lemma shows that
we can restrict the set of algorithm maps under consideration to those
which treat equivalent configurations equivalently. We call algorithm
maps obeying this condition in addition to the restrictions of
Section~\ref{sec:simplifications} \emph{simplified algorithm maps}.
\begin{lem}
  \label{lem:equal-confs} At~$1+O(\epsilon)$ loss we can restrict to
  algorithm maps~$f$ such that $f(C)\sim f(C')$ for any two equivalent
  configurations~$C,C'$.
\end{lem}
\begin{proof}
  Let $f$ be an algorithm map. For each equivalence class
  $\C'\subseteq\C$ of the set of configurations we pick a
  representative $C$ which is realistic for~$f$. For each configuration $C'\in\C'$, we define a new
  algorithm map $\bar{f}$ by setting $\bar{f}(C')$ to be the
  interval-schedule for $C'$ which is equivalent to $f(C)$.
  One can show by induction that $\bar{f}$ is always in a
  configuration such that an equivalent configuration is realistic for
  $f$. Hence, equivalence classes without realistic configurations
  for~$f$ are not relevant. We claim that
  $\rho_{\bar{f}}\le(\eo)\rho_{f}$.

  Consider an instance $\bar{\I}$. We show that there is an
  instance~$\I$ such
  that~$\bar{f}(\bar{\I})/\opt(\bar{\I})\le(\eo)f(\I)/\opt(\I)$ which
  implies the claim. Consider an interval~$I_{\bar{x}}$. Let~$\bar{C}$
  be the end-configuration obtained when~$\bar{f}$ is applied
  iteratively on~$\bar{\I}$. Let~$C$ be the representative of the
  equivalence class of~$\bar{C}$, which was chosen above and which is
  realistic for~$f$ ($C$ is also an end-configuration). Therefore,
  there is an instance~$\I$ such that~$C$ is reached at time~$R_{x}$
  when $f$ is applied on~$\I$. Hence, $I$ is the required instance
  since the relevant jobs dominate the objective value~(see
  Lemma~\ref{lem:irrelevant-jobs}) and~$C \sim \bar{C}$.
\end{proof}
\begin{lem}\label{lem:const-numb-maps}
  There are only constantly many simplified algorithm maps. Each
  simplified algorithm map can be described using finite
  information.
\end{lem}
\begin{proof}
  From the simplifications introduced in Section~\ref{sec:simplifications}
  follows that the domain of the
  algorithm maps under consideration contains only constantly many
  equivalence classes of configurations. Also, the target space
  contains only constantly many equivalence classes of interval-schedules.  For
  an algorithm map $f$ which obeys the restrictions of
  Section~\ref{sec:simplifications}, the interval-schedule $f(C)$ is fully
  specified when knowing only $C$ and the equivalence class which
  contains $f(C)$ (since the irrelevant jobs are moved to their safety
  net anyway). Since $f(C)\sim f(C')$ for a simplified algorithm map
  $f$ if $C\sim C'$, we conclude that there are only constantly many
  simplified algorithm maps. Finally, each equivalence class of
  configurations and interval-schedules can be characterized using only finite
  information, and hence the same holds for each~simplified~algorithm~map.
\end{proof}
The next lemma shows that up to a factor $1+\epsilon$ worst case
instances of simplified algorithm maps span only constantly many
intervals.  Using this property, we will show in the subsequent lemmas
that the competitive ratio of a simplified algorithm map can be
determined algorithmically up to a $1+\epsilon$ factor.
\begin{lem}
  \label{lem:worst-instance-length} There is a constant $\R$ such that
  for any instance~$I$ and any simplified algorithm map~$f$ there is a
  realistic end-configuration~$\tilde{C}$ for an
  interval~$I_{\tilde{x}}$ with~$\tilde{x}\leq\R$ which is equivalent
  to the corresponding end-configuration when~$f$ is applied to~$I$.
\end{lem}
\begin{proof}
  Consider a simplified algorithm map~$f$. For each interval $I_{x}$,
  denote by $\C_{x}^{f}$ the set of realistic equivalence classes for
  $I_{x}$, i.e., the equivalence classes which have a realistic
  representative for $I_{x}$. Since there are constantly many
  equivalence classes and thus constantly many \emph{sets }of 
  equivalence classes, there must be a constant~$\R$ independent of~$f$
  such that~$\C_{\bar{x}}^{f}=\C_{\bar{x}'}^{f}$ for
  some~$\bar{x}<\bar{x}'\le\R$. Since~$f$ is simplified it can be
  shown by induction that $\C_{\bar{x}+k}^{f}=\C_{\bar{x}'+k}^{f}$ for
  any $k\in\mathbb{N}$, i.e.,~$f$ \emph{cycles} with period
  length~$\bar{x}'-\bar{x}$.

  Consider now some instance~$I$ and let~$C$ with interval~$I_{x}$ be
  the corresponding end-configuration when~$f$ is applied to~$I$.
  If~$x\leq\R$ we are done. Otherwise there must be
  some~$k\leq\bar{x}'-\bar{x}$ such
  that~$\C_{\bar{x}+k}^{f}=\C_{x}^{f}$ since~$f$ cycles with this
  period length. Hence, by definition of~$\C_{\bar{x}+k}^{f}$ there
  must be a realistic end-configuration~$\tilde{C}$ which is
  equivalent to~$C$ for the interval~$I_{\tilde{x}}$
  with~$\tilde{x}:=\bar{x}+k\leq\R$. \end{proof}
\begin{lem}
  \label{lem:approx-comp-factor} Let~$f$ be a simplified algorithm
  map. There is an algorithm which approximates~$\rho_{f}$ up to
  a factor~$1+\epsilon$, i.e., it computes a value~$\rho'$
  with~$\rho'\le\rho_{f}\le(1+O(\epsilon))\rho'$.
\end{lem}
\begin{proof}[Proof sketch.]
  By Lemma~\ref{lem:irrelevant-jobs}, the relevant jobs in a
  configuration dominate the entire objective value. In particular, we do
  not need to know the irrelevant jobs of a configuration if we only
  want to approximate its objective value up to a factor of $\eo$. For
  an end-configuration $C$ denote by $val_{C}(J_{\Rel}(C))$ the
  objective value of the jobs in $J_{\Rel}(C)$ in the history of $C$.
  We define $r(C):=val_{C}(J_{\Rel}(C))/\opt(J_{\Rel}(C))$ to be the
  achieved competitive ratio of $C$ when restricted to the relevant
  jobs. According to Lemma~\ref{lem:worst-instance-length}, it
  suffices to construct the sets $\C_{0}^{f},...,\C_{\R}^{f}$ in order
  to approximate the competitive ratio of all end-configurations in
  these sets. We start with $\C_{0}^{f}$ and determine $f(C)$ for one
  representant $C$ of each equivalence class $\C\in\C_{0}^{f}$.  Based
  on this we determine the set $\C_{1}^{f}$. We continue inductively
  to construct all sets $\C_{x}^{f}$ with $x\le\R$.

  We define $r_{\max}$ to be the maximum ratio $r(C)$ for an
  end-configuration $C\in\cup_{0\le x\le\R}\C_{x}^{f}$. Due to
  Lemma~\ref{lem:worst-instance-length} and
  Lemma~\ref{lem:irrelevant-jobs} the value~$r_{\max}$ implies the
  required~$\rho'$ fulfilling the properties claimed in this lemma.
\end{proof}

Our main algorithm works as follows. We first enumerate all simplified
algorithm maps. For each simplified algorithm map $f$ we approximate~$\rho_{f}$
using Lemma~\ref{lem:approx-comp-factor}. We output the map~$f$ with
the minimum~(approximated) competitive ratio.  Note that the
resulting online algorithm has polynomial running time: All
simplifications of a given instance can be done efficiently and for a
given configuration, the equivalence class of the schedule for the
next interval can be found in a look-up table of constant size.

\begin{thm}\label{thm:online-approx-scheme-Pm-pmtn}
  For any $m\in \mathbb{N}$ we obtain a competitive-ratio approximation scheme for
  \abc{Pm}{r_{j},pmtn}{\sum w_{j}C_{j}}.
\end{thm}

\section{Extensions to Other Settings}
\label{sec:Extensions}

\paragraph{Non-preemptive Scheduling.}
\label{sec:nonpmtn}

When preemption is not allowed, the definition of the safety
net~(Lemma~\ref{lem:safety-net}) needs to be adjusted since we cannot
ensure that at the end of each interval~$I_{x+s}$ there is a machine idle.
However, we can guarantee that there is a reserved space
\emph{somewhere} in~$[R_{x},R_{x+s})$ to process the small and big
jobs in~$S_{x} \cup L_{x}$. Furthermore, we cannot enforce that a big
job~$j$ is processed for exactly a certain multiple of~$p_{j}\N$ in
each interval~(Lemma~\ref{lem:large-job-atoms}). To solve this, we
pretend that we could preempt~$j$ and ensure that after~$j$ has been
preempted its machine stays idle until~$j$ continues.
Next, we can no longer assume that each part can be treated independently
(Lemma~\ref{lem:split-the-instance-pmtn}).
Since some of the remaining jobs
at the end of a part
may have already started processing, we cannot simply move them to their
safety net. 
Here we use the following~simplification.

\begin{lem}\label{lem:split-the-instance-non-pmtn}%
  Let~$\first{i}$ denote the job that is released first in
  part~$\Part{i}$.  At~$\e$ loss, we can restrict to instances such
  that~$\sum_{\ell=1}^{i-1} rw(\Part{\ell}) \le
  \frac{\epsilon}{\ee^{s}}\cdot rw(\first{i})$, i.e.,~$\first{i}$ dominates
all previous parts.
\end{lem}

Therefore, at $\e$ loss it is enough to consider only the currently running jobs from the previous part
and the last $\Gamma$ intervals from the current part when taking decisions.
Finally, we add some minor modifications to handle the case that a
currently running job is dominated by some other job.
With these adjustments, we have only constantly many
equivalence classes for interval-schedules and configurations, which allows us
to construct a competitive-ratio approximation scheme as in
Section~\ref{sec:Abstraction-online-algos}.

\begin{thm}\label{thm:non-pmtn-algo}
  For any $m\in \mathbb{N}$ we obtain a competitive-ratio approximation scheme for \abc{Pm}{r_j}{\sum
    w_jC_j}\,.
\end{thm}

\paragraph{Scheduling on Related Machines.}
\label{sec:related}

In this setting, each machine~$i$ has associated a
speed~$s_i$, such that processing job~$j$ on machine~$i$
takes~$p_{j}/s_{i}$ time units. W.l.o.g. the slowest machine has unit
speed.  Let~$s_{\max}$ denote the maximum speed in an instance. An
adjusted version of Lemma~\ref{lem:first-rounding} ensures that at
$1+\epsilon$ loss~$r_{j} \ge \epsilon\, p_{j}/s_{\max}$ for all jobs
$j$~(rather than~$r_{j} \ge \epsilon p_{j}$). Furthermore, we can
bound the number of distinct processing times and the number of
released jobs of each interval, using similar arguments as in the
unit-speed~case.

\begin{lem}
  \label{lem:related-volume-bound}
  At $1+O(\epsilon)$ loss we can restrict to instances where for each
  release date the number of released jobs and the number of distinct
  processing times is bounded by a constant depending only
  on~$\epsilon$,~$m$, and~$s_{\max}$.
\end{lem}

We establish the safety net for the jobs of each release date $R_{x}$
only on the {\em fastest} machine and thereby ensure the condition of
Lemma~\ref{lem:safety-net} in the related machine setting.
For the non-preemptive setting we
incorporate the adjustments introduced in Section~\ref{sec:nonpmtn}.
Since at $1+\eps$ loss we can round the speeds of the machines to powers of $1+\eps$
we obtain the following result.

\begin{thm}
  For any $m\in \mathbb{N}$ we obtain competitive-ratio approximation schemes for \abc{Qm}{r_j, pmtn}{\sum
    w_jC_j} and \abc{Qm}{r_j}{\sum w_jC_j}, assuming that the speeds
  of any two machines differ by at most a constant factor.
\end{thm}

In the preemptive setting we can strengthen the result and give a
competitive-ratio approximation scheme for the case that machine speeds are part
of the input, that is, we obtain a nearly optimal competitive ratio
for {\em any} speed vector. The key is to bound the variety of
different speeds. To that end, we show that at~$\e$ loss a very fast
machine can simulate~$m-1$ very slow machines.

\begin{lem}\label{lem:Qm-pmtn-speed-bound}
  For \abc{Qm}{r_j, pmtn}{\sum w_jC_j}, at $1+O(\epsilon)$ loss, we
  can restrict to instances in which $s_{\max}$ is bounded by
  $m/\epsilon$.
\end{lem}

\noindent As speeds are geometrically rounded, we have for each
value~$m$ only finitely many speed vectors. Thus, our enumeration
scheme finds a nearly optimal online algorithm with a particular
routine for each speed~vector.

\begin{thm}
  For any $m\in \mathbb{N}$ we obtain a competitive-ratio approximation scheme for \abc{Qm}{r_j,pmtn}{\sum
    w_jC_j}\,.
\end{thm}

\paragraph{Preemptive Scheduling on Unrelated Machines.}

When each job~$j$ has its individual processing time~$p_{ij}$ on
machine~$i$, the problem complexity increases significantly. We
restrict to preemptive scheduling and show how to decrease the
complexity to apply our approximation scheme. The key is to bound the
range of the finite processing times for each job (which is
unfortunately not possible in the non-preemptive case,
see~\cite{afrati99} for a counterexample).

\begin{lem}\label{lem:Rm-range-processing-times}
  At~$\e$ loss we can %
  restrict to instances in which for each
  job~$j$ the ratio of any two of its finite processing times is bounded
  by~$m/\eps$.
\end{lem}

The above lemma allows us to introduce the notion of \emph{job
  classes}. Two jobs~$j,j'$ are of the same class if they have finite
processing times on exactly the same machines and
~$p_{ij}/p_{ij'}=p_{i'j}/p_{i'j'}$ for any two such machines~$i$
and~$i'$.  For fixed $m$, the number of different job classes is
bounded by a constant $W$.

For each job class, we define large and small tasks similar to the
identical machine case: for each job~$j$ we define a
value~$\tilde{p}_{j}:=\max_{i}\{p_{ij}|p_{ij}<\infty\}$ and say a job
is \emph{large} if~
$\tilde{p}_{j}\ge \epsilon^{2}r_{j}/W$
and
\emph{small} otherwise. For each job class separately, we perform the
adjustments
of Section~\ref{sec:simplifications}. This yields the following lemma.

\begin{lem}
  At $1+O(\epsilon)$ loss we can restrict to instances and schedules such that
  \begin{compactitem}
  \item for each job class, the number of distinct values
    $\tilde{p}_j$ of jobs $j$ with the same release date is bounded by
    a constant,
  \item for each job class, the number of jobs with the same release
    date is bounded by a constant $\tilde{\Delta}$,
  \item a large job~$j$ is only preempted at integer multiples of
    $\tilde{p}_{j} \cdot \tilde{\mu}$ for some constant~$\tilde{\mu}$
    and small jobs are never preempted and finish in the same interval
    where they start.
  \end{compactitem}
\end{lem}

The above lemmas imply that both, the number of equivalence classes of
configurations and the number of equivalence classes for
interval-schedules are bounded by constants. Thus, we can apply
the enumeration scheme from Section~\ref{sec:Abstraction-online-algos}.

\begin{thm}
 For any $m\in \mathbb{N}$ we obtain a competitive-ratio approximation scheme for \abc{Rm}{r_j,
    pmtn}{\sum w_jC_j}\,.
\end{thm}

\section{Randomized algorithms} \label{sec:randomized}

When algorithms are allowed to make random choices and we consider
expected values of schedules, we can restrict to instances which span
only constantly many periods. Assuming the simplifications of
Section~\ref{sec:simplifications}, this allows a restriction to
instances with a constant number of jobs.

\begin{lem}\label{lem:rand-const-periods}
  For randomized algorithms, at~$\e$ loss we can restrict to instances
  in which all jobs are released in at most~$\ee^{s}/\eps$
  consecutive periods.
\end{lem}
\begin{proof}[Proof idea.] 
Beginning at a randomly chosen period~$Q_i$
  with~$i\in [0,M)$, with~$M:=\lceil\ee^{s}/\eps\rceil$, we move
  all jobs released in~$Q_{i+kM}$,~$k=0,1,\ldots$, to their safety
  net. At~$\e$ loss, this gives us a partition into parts, at the end
  of which no job remains, and we can treat each part independently.
\end{proof}

A randomized online algorithm can be viewed as a function that maps
every possible configuration~$C$ to a probability distribution of
interval-schedules which are feasible for~$C$. To apply our
algorithmic 
framework from the deterministic setting
that enumerates
all algorithm maps, we discretize the probability space and define
discretized algorithm maps. To this end, let~$\bar{\Gamma}$ denote the
maximum number of intervals in instances with at~most~$\ee^{s}/\eps$~periods.
\begin{defn}[Discretized algorithm maps]
  Let~$\bar{\C}$ be the set of configurations for intervals~$I_{x}$
  with~$x\le\bar{\Gamma}$, let~$\bar{\mathcal{S}}$ be the set of
  interval-schedules for intervals~$I_{x}$ with~$x\le\bar{\Gamma}$,
  and let~$\delta>0$. A \emph{$\delta$-discretized algorithm map }is a
  function~$f:\bar{\C}\times\bar{\mathcal{S}}\rightarrow[0,1]$ such
  that
  $f(C,S)=k\cdot \delta$ with some~$k\in\mathbb{N}_{0}$ for all
  ~$C\in\bar{\C}$ and~$S\in\bar{\mathcal{S}}$, and
  $\sum_{S\in\bar{\mathcal{S}}}f(C,S)=1$ for all ~$C\in\bar{\C}$.
\end{defn}
By restricting to $\delta$-discretized
algorithm maps we do not lose too much in the competitive ratio.

\begin{lem}\label{lem:discretize-prob}
  There is a value~$\delta>0$ such that for any (randomized) algorithm
  map~$f$ there is a~$\delta$-discretized randomized algorithm map~$g$
  with~$\rho_{g}\le\rho_{f}\ee$.
\end{lem}
\begin{proof}[Proof idea.] 
Let~$f$ be a randomized algorithm map and
  let~$\delta>0$ such that $1/\delta \in \mathbb{N}$. We define a
  new~$\delta$-discretized algorithm map~$g$. For each
  configuration~$C$ we define the values~$g(C,S)$ such
  that~$\left\lfloor f(C,S)/\delta \right\rfloor \cdot\delta\le
  g(C,S)\le\left\lceil f(C,S)/\delta\right\rceil \cdot\delta$
  and~$\sum_{S\in\mathcal{S}}g(C,S)=1$. To see
  that~$\rho_{g}\le\ee\rho_{f}$, consider an instance~$\I$ and a
  possible schedule~$S(\I)$ for $\I$. There is a probability~$p$
  that~$f$ outputs~$S(\I)$. We show that the schedules which have large
  probability~$p$ dominate~$\Exp{f(\I)}/\opt(\I)$. We show further
  that if $p$ is sufficiently large, the probability that~$g$
  produces~$S(\I)$ is in~$[p/(1+\eps),p(1+\eps))$, which implies the
  Lemma.
\end{proof}

Like in the deterministic case, we can now show that at $1+\eps$ loss
it suffices to restrict to \emph{simplified $\delta$-discretized
  algorithm maps} which treat equivalent configurations equivalently,
similar to Lemma~\ref{lem:equal-confs} (replacing~$\Gamma$
by~$\bar{\Gamma}$ in Definition~\ref{def:irrelevant_jobs} of the
irrelevant jobs).
As there are only constantly many of these maps, we enumerate all of
them, test each map for its competitive ratio, and select the best of
them.

\begin{thm} We obtain randomized competitive-ratio approximation schemes for \abc{Pm}{r_j,
    (pmtn)}{\sum w_jC_j},
      \abc{Qm}{r_j, pmtn}{\sum w_jC_j}, and
  \abc{Rm}{r_j, pmtn}{\sum w_jC_j}\, and for \abc{Qm}{r_j}{\sum
    w_jC_j} with a bounded range of speeds for any fixed $m\in\mathbb{N}$.
\end{thm}

\section{General Min-Sum Objectives and Makespan}
\label{sec:other-objectives}
\vspace{-0.5ex}

In this section we briefly argue how the techniques presented above for minimizing
$\sum_{j}w_{j}C_{j}$ can be used for constructing online algorithm
schemes for other scheduling problems with jobs arriving online over time, namely for minimizing $\sum_{j\in J}w_jf(C_{j})$, with~$f(x)=k\cdot x^{\alpha}$ and constant~$\alpha\geq 1, k>0$, and the makespan.

Since monomial functions~$f$ have the property that~$f((\e) C_j)\leq (\eo) f(C_j)$,
the arguments in previous sections apply almost directly to the generalized min-sum objective.
In each step of simplification and abstraction, we have an increased loss in the performance guarantee, but it is covered~by~the~$O(\eps)$-term.

Consider now the makespan objective. The {\em simplifications within intervals} of Section~\ref{sec:simplifications} are based on arguing on completion times of individual jobs, and clearly hold also for the last job. Thus, they directly apply to makespan minimization. We simplify the definition of {\em irrelevant history} by omitting the partition of the instance into parts and we define a job $j$ to be \emph{irrelevant at time $R_{x}$ }if $r_{j}\le R_{x-s}$.

Based on this definition, we define equivalence
classes of configurations~(ignoring weights and previous interval-schedules)~and again restrict
to algorithm maps~$f$ with $f(C)\sim f(C')$ for any two equivalent
configurations~(Lemma~3.6).
Lemmas 3.7--3.9 then hold accordingly and yield a competitive-ratio approximation scheme. 
Finally, the adjustments of Sections~\ref{sec:Extensions} and~\ref{sec:randomized} can be made 
accordingly. Without the partition of the instance into parts, 
  this even becomes easier in the non-preemptive setting. 
  Thus, we can state the following result.
\begin{thm}
    For any $m \in \mathbb{N}$ there are deterministic and randomized competitive-ratio approximation schemes for preemptive and non-preemptive scheduling, on $m$ identical, related~(with bounded speed ratio when non-preemptive), and unrelated machines~(only preemptive) for the objectives of minimizing~$C_{\max}$ and minimizing~$\sum_{j\in J}w_j f(C_j)$, with~$f(x)=k\cdot x^{\alpha}$ and constant~$\alpha\geq 1,k>0$.
\end{thm}

\section{Conclusions}
We introduce the concept of competitive-ratio approximation schemes
that compute online algorithms with a competitive ratio arbitrarily
close to the best possible competitive ratio. We provide such schemes
for various problem variants of scheduling jobs online to minimize
the weighted sum of completion times, arbitrary monomial cost
functions, and the makespan.

The techniques derived in this paper provide a new and interesting view on the
behavior of online algorithms. We believe that they contribute to the
understanding of such algorithms and possibly open a new line of
research in which they yield even further insights. In particular, it
seems promising that our methods could also be applied to other online
problems than scheduling jobs arriving online over time.

\newpage

\bibliography{scheduling}
\bibliographystyle{abbrv}
\newpage
\appendix

\section{Related work}\label{app:related_work}

\paragraph{Sum of weighted completion times.} The offline variants of nearly all
problems under consideration are \npc-hard.  This is true already for the special case of
a single
machine~\cite{labetoulle84,lenstra77}. Two restricted
 variants can be solved optimally in polynomial time. {\em Smith's
   Rule} solves the
 problem~\abc{1}{}{\sum w_jC_j} to optimality by scheduling jobs in non-increasing order
 of weight-to-processing-time ratios~\cite{smith56}. Furthermore,
 scheduling by shortest remaining processing times yields an optimal
 schedule for~\abc{1}{r_j,pmtn}{\sum w_jC_j}~\cite{schrage68}.
However, for the other settings polynomial-time approximation schemes have been
developed~\cite{afrati99}, even when the number of machines is part of
the input.

The online setting has been a highly active field of research in the
past fifteen years. A whole sequence of papers appeared introducing
new algorithms, new relaxations and analytical techniques that
decreased the gaps between lower and upper bounds on the optimal
competitive
ratio~\cite{goemans02,schulzS02,hallSSW97,sitters10-ipco,sitters10-orl,andersonP04,
  chekuriMNS01,hoogeveenV96,correaW09,goemans97,megowS04,megow06-diss,chungNS10,
  schulzS02-unrelated,liuL09,luSS03,stougieV02,phillipsST98,chakrabartiPSSSW96,
  seiden00,epsteinvS01}.  We do not intend to give a detailed history
of developments; instead, we refer the reader to %
overviews, e.g., in~\cite{megow06-diss,correaW09}.
Table~\ref{tab:results} summarizes the current
state-of-the-art on best known lower and upper bounds on the optimal competitive ratios.
Interestingly, despite the considerable effort, optimal competitive
ratios are known only for~\abc{1}{r_j, pmtn}{\sum
  C_j}~\cite{schrage68} and for non-preemptive single-machine
scheduling~\cite{andersonP04,stougieV02,hoogeveenV96,chekuriMNS01}.
In all other scheduling settings remain unsatisfactory, even quite significant gaps.

\begin{table}[htb]
{\footnotesize
\begin{tabularx}{\linewidth}{l@{\hspace*{9ex}}X@{\hspace*{9ex}}X@{\hspace*{9ex}}X@{\hspace*{9ex}}X}
  \toprule
  & \multicolumn{2}{c}{deterministic\hspace*{1.5cm}} & \multicolumn{2}{c}{randomized\hspace*{.5cm}}\\
  problem & lower bounds & upper bounds & lower bounds & upper bounds\\
  \midrule
  \abc{1}{r_j, pmtn}{\sum C_j} & $1$ & $1$ \hfill \cite{schrage68} & $1$ & $1$ \hfill \cite{schrage68}\\
  \abc{1}{r_j, pmtn}{\sum w_jC_j} & $1.073$ \hfill \cite{epsteinvS01} & $1.566$ \hfill \cite{sitters10-orl} & $1.038$ \hfill \cite{epsteinvS01} & $4/3$ \hfill \cite{schulzS02}\\
  \abc{1}{r_j}{\sum C_j} & $2$ \hfill \cite{hoogeveenV96} & $2$ \hfill \cite{hoogeveenV96} & $\frac{e}{e-1}\approx 1.58\ \ $ \hfill \cite{stougieV02} & $\frac{e}{e-1}$ \hfill \cite{chekuriMNS01}\\
  \abc{1}{r_j}{\sum w_jC_j} & $2$ \hfill \cite{hoogeveenV96} & $2$ \hfill \cite{andersonP04} & $\frac{e}{e-1}$ \hfill \cite{stougieV02} & $1.686$ \hfill \cite{goemans02}\\
  \midrule
  \abc{P}{r_j, pmtn}{\sum C_j} & $1.047$ \hfill \cite{hoogeveenV96}& $5/4$ \hfill \cite{sitters10-ipco} & $1$ & $5/4$ \hfill \cite{sitters10-ipco}\\
  \abc{P}{r_j, pmtn}{\sum w_jC_j} &  $1.047$ \hfill  \cite{vestjens97diss} & $1.791$ \hfill \cite{sitters10-ipco} & $1$ & $1.791$ \hfill \cite{sitters10-ipco}\\
  \abc{P}{r_j}{\sum w_jC_j} & $1.309$ \footnotemark[1]
  \hfill \cite{vestjens97diss} & $1.791$ \hfill \cite{sitters10-ipco} & $1.157$ \hfill \cite{seiden00} & $1.791$ \hfill \cite{sitters10-ipco}\\
  \midrule
  \abc{R}{r_j}{\sum w_jC_j} & $1.309$ \hfill \cite{vestjens97diss}  & $8$ \hfill \cite{hallSSW97} & $1.157$ \hfill \cite{seiden00} & $8$ \hfill \cite{hallSSW97}\\
  \bottomrule
\end{tabularx}
\caption{\small %
Lower and upper bounds on the %
competitive ratio for deterministic and randomized online algorithms.}
\label{tab:results}
}
\end{table}

\footnotetext[1]{For $m=1,2,3,4,5, \dots 100$ the lower bound is
   $LB=2, 1.520, 1.414, 1.373, 1.364, \dots 1.312$.}

\paragraph{More general min-sum (completion time) objectives.} Recently, there has been an increasing interest in studying generalized cost functions. So far, this research has focussed on offline problems. The most general case is when each job may have its individual non-decreasing cost function~$f_j$. For scheduling on a single machine with release dates and preemption,~$1|r_j,pmtn|\sum f_j$, Bansal and Pruhs~\cite{bansalP10} gave a randomized~$\OO(\log \log (nP))$-approximation, where~$P=\max_{j\in J}p_j$. In the case that all jobs have identical release dates, the approximation factor reduces to~$16$. Cheung and Shmoys~\cite{cheungS11} improved this latter result and gave a deterministic~$(2+\eps)$-approximation. This result applies also on a machine of varying speed.

The more restricted problem with a global cost function~$1|r_j,pmtn|\sum w_jf(C_j)$ has been studied by Epstein~et al.~\cite{epsteinLMMMSS12} in the context of universal solutions. They gave an algorithm that produces for any job instance one scheduling solution that is a~$(4+\eps)$-approximation for any cost function and even under unreliable machine behavior. H{\"o}hn and Jacobs~\cite{hoehnJ12} studied the same problem without release dates. They analyzed the performance of Smith's Rule~\cite{smith56} and gave tight approximation guarantees for all convex and all concave functions~$f$.

\paragraph{Makespan.} The online makespan minimization problem has been extensively studied in a different online paradigm where jobs arrive one by one (see \cite{FleischerWahl2000,RudinChandrasekaran2003} and references therein). Our model, in which jobs arrive online over time, is much less studied. In the identical parallel machine environment, Chen and Vestjens~\cite{chenV97} give nearly tight bounds on the optimal competitive ratio,~$1.347 \leq \rho^* \leq 3/2$, using a natural online variant of the well-known largest processing time first algorithm.

In the offline setting, polynomial time approximation schemes are known for identical \cite{HochbaumShmoys87} and uniform machines \cite{HochbaumShmoys88}. For unrelated machines, the problem is NP-hard to approximate with a better ratio than~$3/2$ and a $2$-approximation is known \cite{LenstraShmoysTardos90}. If the number of machines is bounded by a constant there is a~PTAS~\cite{LenstraShmoysTardos90}.
\section{Proofs of Section~\ref{sec:simplifications}} %

First, we will show that the number of distinct processing
times of large jobs in each interval can be upper-bounded by a constant.
To achieve this, we partition the jobs of an instance into large and small jobs.
With respect to a release date $R_{x}$ we say that a job $j$ with $r_{j}=R_{x}$
is \emph{large} if $p_{j}\ge\epsilon^{2}I_{x}=\epsilon^{3}R_{x}$
and \emph{small} otherwise.
Abusing notation, we refer to~$|I_{x}|$ also by~$I_{x}$. Note that $I_x=\eps\cdot(1+\eps)^x$.
\begin{lem}
\label{lem:sizes-large-jobs}
The number of distinct processing times
of jobs in each set $L_x$ is bounded by $4\log_{\e}\frac{1}{\varepsilon}$.\end{lem}

\begin{proof}%
  For any $j\in L_{x}$ the processing time $p_{j}$ is a power of
  $1+\epsilon$, say $p_{j}=(1+\epsilon)^{y}$. Hence, we have that
  $\varepsilon^{3}\ee^{x}<\ee^{y}\le\frac{1}{\varepsilon}\ee^{x}$.
  The number of integers $y$ which satisfy the above inequalities is
  bounded by $4\log_{\e}\frac{1}{\varepsilon}$, which yields the
  constant claimed in the lemma.
\end{proof}

Furthermore, we can bound the number of large jobs of each job size
which are released at the same time.

\begin{lem}
\label{lem:number-large-jobs}Without loss, we can restrict to instances with $|L_{x}|\le(m/\epsilon^{2}+m)4\log_{\e}\frac{1}{\varepsilon}$
for each set~$L_{x}$.\end{lem}

\begin{proof}%
  Let $L_{x,p}\subseteq L_{x}$ denote the set of jobs in $L_{x}$ with
  processing time $p$. By an exchange argument, one can restrict to schedules
  such that at each point in time at most $m$ jobs in
  $L_{x,p}$ are partially (i.e., to some extent but not completely)
  processed.  Since $p_{j}\ge\epsilon^{2}I_{x}$ for each job $j\in
  L_{x}$, at most $m/\epsilon^{2}+m$ jobs in $L_{x,p}$ are touched
  within $I_{x}$.  By an exchange argument we can assume that they are
  the $m/\epsilon^{2}+m$ jobs with the largest weight in $L_{x,p}$.
  Hence, the release date of all other jobs in $L_{x,p}$ can be moved
  to $R_{x+1}$ without any cost. Since due to
  Lemma~\ref{lem:sizes-large-jobs} there are at most
  $4\log_{\e}\frac{1}{\varepsilon}$ distinct processing times $p$ of
  large jobs in $L_{x}$, the claim follows.
\end{proof}

We now just need to take care of the small jobs.
Denote by $w_{j}/p_{j}$ the \emph{Smith's ratio} of a job $j$.
An ordering where the jobs are ordered non-increasingly by their Smith's
ratios is an ordering according to \emph{Smith's rule}. The next lemma
shows that scheduling the small jobs according to Smith's Rule is
almost optimal and small jobs do not even need to be preempted or
to cross intervals. For a set of jobs $J$ we define $p(J):=\sum_{j\in J}p_{j}$.

\begin{lem}
\label{lem:smiths-rule}At $\e$ loss we can restrict to schedules such that
for each interval $I_{x}$ the small jobs scheduled within this interval
are chosen by Smith's Rule from the set $\bigcup_{x'\le x}S_{x'}$, no small job is preempted,
    any small job finishes in the same interval where it started and $p(S_{\!x})\le m\cdot I_{x}$ for each interval $I_{x}$.
\end{lem}
\begin{proof}%
  By an exchange argument one can show that it is optimal to schedule
  the small jobs by Smith's Rule if they can be arbitrarily divided
  into smaller jobs (where the weight is divided proportional to the
  processing time of the smaller jobs). Start with such a schedule and
  stretch time once. The gained free space is enough to finish
  all small jobs which are partially scheduled in each interval.

  For the last claim of the lemma, note that the total
  processing time in each interval $I_{x}$ is $mI_{x}$. Order the
  small jobs non-increasingly by their Smith's Ratios and pick them
  until the total processing time of picked jobs just does not exceed
  $mI_{x}$. The release date of all other jobs in $S_{x}$ can be
  safely moved to $R_{x+1}$ since due to our modifications we would
  not schedule them in $I_{x}$ anyway.
\end{proof}

\begin{lem-safety-net}[restated]%
  There is a constant $s$ such that at $\eo$
  loss we can restrict to schedules such that for each interval
  $I_{x}$ there is a subinterval of $I_{x+s-1}$ %
  which is large enough to process all jobs released at~$R_x$ and during which only jobs in~$R_x$ are executed.
  We call this subinterval the \emph{safety net} of
  interval $I_{x}$. We can assume that each job released at~$R_x$ finishes before time~$R_{x+s}$.
\end{lem-safety-net}
\begin{proof}%
  By Lemmas \ref{lem:smiths-rule} and $\ref{lem:number-large-jobs}$ we bound $p(S_{x})+p(L_{x})$ by

  \begin{eqnarray*}
    p(S_{x})+p(L_{x}) & \le & m\cdot I_{x}+(m/\epsilon^{2}+m) \cdot
    \left(4\log_{\e}\frac{1}{\varepsilon}\right) \cdot\frac{1}{\varepsilon}\ee^{x}\\
    & \le & m\cdot\ee^{x}\left(\epsilon+\frac{8}{\epsilon^{3}}\log_{\e}\frac{1}{\varepsilon}\right)\\
    & = & \varepsilon\cdot I_{x+s-1}
  \end{eqnarray*}
  for a suitable constant $s$, depending on $\epsilon$ and $m$. Stretching time once, we gain enough free space at the end of each
  interval $I_{x+s-1}$ to establish the safety net for each job set
  $p(S_{x})+p(L_{x})$.%

  {}
\end{proof}

\begin{lem}\label{lem:tiny-jobs}
  There is a constant $d$ such that %
  we can at $\eo$ loss restrict to instances such that
  $p_{j}>\frac{\epsilon}{2d}\cdot I_{x}$ for each job $j \in S_x \cup L_x$.
\end{lem}
\begin{proof}%
We call a job $j$ \emph{tiny} if $p_j \le \frac{\epsilon}{2d}\cdot I_{x}$.
  Let $T_{x}=\{j_{1},j_{2},...,j_{|T_{x}|}\}$ denote all tiny jobs
  released at $R_{x}$. W.l.o.g.~assume that they are ordered
  non-increasingly by their Smith's Ratios $w_{j}/p_{j}$. Let $\ell$
  be the largest integer such that
  $\sum_{i=1}^{\ell}p_{i}\le\frac{\epsilon}{d}\cdot I_{x}$.  We define
  the pack $P_{x}^{1}:=\{j_{1},...,j_{\ell}\}$. We denote by
  $\sum_{i=1}^{\ell}p_{i}$ the processing time of pack $P_{x}^{1}$ and
  by $\sum_{i=1}^{\ell}w_{i}$ its weight. We continue iteratively
  until we assigned all tiny jobs to packs. By definition of the
  processing time of tiny jobs, the processing time of all but
  possibly the last pack released at time $R_{x}$ is in the interval
  $[\frac{\epsilon}{2d}\cdot I_{x},\frac{\epsilon}{d}\cdot I_{x}]$.

  Using timestretching, we can show that at $1+O(\epsilon)$ loss all
  tiny jobs of the same pack are scheduled in the same interval on the
  same machine. Here we use that in any schedule obeying Smith's Rule
and using the safety net (see Lemma~\ref{lem:safety-net})
in each interval there is at most one partially but unfinished pack from
each of at most $s$ previous release dates.
Hence, we can treat the packs as single jobs whose
  processing time and weight matches the respective values of the
  packs. Also, at $\e$ loss we can ensure that also the very last pack
  has a processing time in $[\frac{\epsilon}{2d}\cdot
  I_{x},\frac{\epsilon}{d}\cdot I_{x}]$.  Finally, at $1+O(\epsilon)$
  loss we can ensure that the processing times and weights of the new
  jobs (which replace the packs) are powers of $\e$.
\end{proof}

\begin{lem}\label{lem:sizes-small-jobs}
Assume that there is a constant $d$ such that
$p_{j}>\frac{\epsilon}{2d}\cdot I_{x}$ for each job $j \in S_x$.
Then at $\eo$ loss, the number of distinct processing times of jobs each set $S_x$ is upper-bounded by $(\log_{1+\eps}\eps\cdot 2d)$.
\end{lem}
\begin{proof}%
From the previous lemmas, we have $$\frac{e^2}{2d}\cdot (1+\eps)^x<(1+\eps)^y<\eps^3(1+\eps)^x.$$
The number of integers $y$ satisfying these inequalities is upper-bounded by the claimed constant.
\end{proof}

Lemma \ref{lem:constant-num-proc-times} now follows from the lemmas \ref{lem:sizes-large-jobs} and \ref{lem:sizes-small-jobs}.
Lemma \ref{lem:constant-num-jobs-released} follows from lemmas \ref{lem:smiths-rule}, \ref{lem:tiny-jobs} and \ref{lem:number-large-jobs}.
Next, we prove Lemma \ref{lem:large-job-atoms}:

\begin{proof}[Proof of Lemma \ref{lem:large-job-atoms}]
  The claim about the number of partially processed jobs of each type can be
  assumed without any loss. For the extent of processing, note that due to Lemmas~\ref{lem:number-large-jobs},
  \ref{lem:safety-net}, and \ref{lem:tiny-jobs} there is a
  constant $c$ such that at each time $R_{x}$ the total processing
  time of unfinished large jobs is bounded by $c\cdot R_{x}$. We
  stretch time once. The gained space is sufficient to schedule
  $p_{j}\cdot\N$ processing units of each unfinished large job $j$
  (for an appropriately chosen universal constant $\N$). This allows
  us to enforce the claim.
  The claim about the non-preemptive behavior of small jobs follows from Lemma~\ref{lem:smiths-rule}.
\end{proof}

\begin{proof}[Proof of Lemma \ref{lem:dominated-periods}]
In any schedule the jobs in $\cup_{i=0}^{p-1}Q_{k+i}$ contribute at
least $\sum_{i=0}^{p-1}rw(Q_{k+i})$ towards the objective. If we move
all jobs in $Q_{k+p}$ to their safety nets, they contribute at most

\begin{eqnarray*}
\sum_{j\in Q_{k+p}}r_{j}\ee^{s}\cdot w_{j} & = & \ee^{s}\cdot rw(Q_{k+p})\\
 & \le & \epsilon\cdot\sum_{i=0}^{p-1}rw(Q_{k+i})\\
 & \le & \epsilon\cdot OPT\end{eqnarray*}
to the objective.
\end{proof}

\begin{proof}[Proof of Lemma \ref{lem:split-the-instance-pmtn}]
  We modify a given online algorithm such that each part is treated as
  a separate instance. To bound the cost in the competitive ratio, we
  show that
$\frac{\alg(\I)}{\opt(\I)}\le\max_{i}\{\frac{\alg(\Part{i})}{\opt(\Part{i})}\}
\cdot(1+O(\epsilon)).$
  By the above lemmas, there is a $(1+O(\epsilon))$-approximative
  (offline) solution in which at the end of each part $\Part{i}$ each job
  has either completed or has been moved to its safety net. Denote
  this solution by $\opt'(I)$ and by $\opt'(\Part{i})$ its respective part
  for each part $\Part{i}$. Note that $\opt'(\I)=\sum_{i}\opt'(\Part{i})$. Then,

 $$\frac{\sum_{i=1}^{k}\alg(\Part{i}))}{\opt(\I)} \le
\frac{\sum_{i=1}^{k}\alg(\Part{i}))}{\sum_{i=1}^{k}\opt'(\Part{i})}
\cdot(1+O(\epsilon)) \le
\max_{i=1,...,k}\{\frac{\alg(\Part{i})}{\opt(\Part{i})}\}\cdot(1+O(\epsilon)).$$

\end{proof}

\begin{proof}[Proof of Lemma \ref{lem:K-periods}]
We show that $\ee^{s}\sum_{i=1}^{p-K-1}rw(Q_{i})<\epsilon\cdot\sum_{i=p-K}^{p}rw(Q_{i})$
for a sufficiently large value~$K$. This will then be the claimed constant.
Let $\delta':=\frac{\epsilon}{\ee^{s}}$.
By assumption, we have that
$rw(Q_{i+1})>\delta'\cdot\sum_{\ell=1}^{i}rw(Q_{\ell})$ for each~$i$.
This implies that $\frac{rw(Q_{i+1})}{\sum_{\ell=1}^{i+1}rw(Q_{\ell})}>\frac{\delta'}{1+\delta'}$
for each~$i$. Hence,

\[
\frac{\sum_{\ell=1}^{i}rw(Q_{\ell})}{rw(Q_{i+1})+\sum_{\ell=1}^{i}rw(Q_{\ell})}\le1-\frac{\delta'}{1+\delta'}<1\]
for each $i$ and hence,

\[
\sum_{\ell=1}^{i}rw(Q_{\ell})\le(1-\frac{\delta'}{1+\delta'})\sum_{\ell=1}^{i+1}rw(Q_{\ell}).\]
In other words, if we remove $Q_{i+1}$ from $\cup_{\ell=1}^{i+1}Q_{\ell}$
then the total release weight of the set decreases by a factor of
at least $1-\delta'/(1+\delta')<1$. For any $K$ this implies that
\[
\sum_{i=1}^{p-K-1}rw(Q_{i})<\left(1-\frac{\delta'}{1+\delta'}\right)^{K}\sum_{\ell=1}^{p}rw(Q_{\ell})\]
 and hence \[
\sum_{i=1}^{p-K-1}rw(Q_{i})<\frac{1}{1-\left(1-\frac{\delta'}{1+\delta'}\right)^{K}}\left(1-\frac{\delta'}{1+\delta'}\right)^{K}\sum_{i=p-K}^{p}rw(Q_{i}).\]
By choosing $K$ sufficiently large, the claim follows.
\end{proof}

\begin{proof}[Proof of Lemma \ref{lem:irrelevant-jobs}]

We partition $\Ir_{x}(J)$ into two groups: $\Ir_{x}^{\mathrm{old}}(J):=\{j\in\Ir_{x}(J)|r_{j}<R_{x-\Gamma}\}$
and $\Ir_{x}^{\mathrm{new}}(J):=\{j\in\Ir_{x}(J)|r_{j}\ge R_{x-\Gamma}\}$.
Lemma~\ref{lem:K-periods} implies that \begin{equation}
(1+\epsilon)^{s}rw(\Ir_{x}^{\mathrm{old}}(J))\le\epsilon\cdot rw(\Ir_{x}^{\mathrm{new}}(J)\cup\Rel_{x}(J))\label{eq:IR_x_old}\end{equation}
 (recall that the former value is an upper bound on the total weighted
completion time of the jobs in $\Ir_{x}^{\mathrm{old}}(J)$). For
every job $j\in\Ir_{x}^{\mathrm{new}}(J)$ there must be a job $j'\in\Ir_{x}^{\mathrm{old}}(J)\cup\Rel_{x}(J)$
such that $w_{j}<\frac{\epsilon}{\Delta\Gamma\cdot(1+\epsilon)^{\Gamma+s}}w_{j'}$.
We say that such a job $j'$ \emph{dominates} $j$. At most $\Delta$
jobs are released at the beginning of each interval and hence $|\Ir_{x}^{\mathrm{new}}(J)|\le\Delta\Gamma$.
In particular, if $\mathrm{dom}(j')$ denotes all jobs in $\Ir_{x}^{\mathrm{new}}(J)$
which are dominated by $j'$ then

\[
\sum_{j\in\mathrm{dom}(j')}w_{j}r_{j}\le\Delta\Gamma\frac{\epsilon}{\Delta\Gamma\cdot(1+\epsilon)^{\Gamma+s}}w_{j'}r_{j'}\cdot(1+\epsilon)^{\Gamma}\]
This implies that

\begin{eqnarray*}
(1+\epsilon)^{s}rw(\Ir_{x}^{\mathrm{new}}(J)) & \le & (1+\epsilon)^{s}\sum_{j\in\Ir_{x}^{\mathrm{new}}(J)}w_{j}r_{j}\\
 & \le & \sum_{j'\in\Ir_{x}^{\mathrm{old}}(J)\cup\Rel_{x}(J)}\Delta\Gamma\frac{\epsilon}{\Delta\Gamma\cdot(1+\epsilon)^{\Gamma+s}}w_{j'}r_{j'}\cdot(1+\epsilon)^{\Gamma+s}\\
 & \le & \epsilon\cdot\sum_{j'\in\Ir_{x}^{\mathrm{old}}(J)\cup\Rel_{x}(J)}w_{j'}r_{j'}\\
 & = & \epsilon\cdot(rw(\Ir_{x}^{\mathrm{old}}(J))+rw(\Rel_{x}(J)))\end{eqnarray*}
Together with Inequality~\ref{eq:IR_x_old} this implies that

\begin{eqnarray*}
(1+\epsilon)^{s}rw(\Ir_{x}(J)) & = & (1+\epsilon)^{s}(rw(\Ir_{x}^{\mathrm{new}}(J)+rw(\Ir_{x}^{\mathrm{old}}(J)))\\
 & \le & \left(\epsilon\cdot(rw(\Ir_{x}^{\mathrm{old}}(J))+rw(\Rel_{x}(J))\right)+\left(\epsilon\cdot rw(\Ir_{x}^{\mathrm{new}}(J)\cup\Rel_{x}(J))\right)\\
 & \le & 2\epsilon\cdot rw(\Rel_{x}(J))+\epsilon(rw(\Ir_{x}(J))\end{eqnarray*}
and the latter inequality implies that

\begin{eqnarray*}
\sum_{j\in\Ir_{x}(J)}C_{j}w_{j} & \le & (1+\epsilon)^{s}rw(\Ir_{x}(J))\\
 & \le & 2\epsilon\frac{(1+\epsilon)^{s}}{(1+\epsilon)^{s}-\epsilon}rw(\Rel_{x}(J))\\
 & \le & 3\epsilon\cdot rw(\Rel_{x}(J))\end{eqnarray*}

\end{proof}

\section{Proofs of Section~\ref{sec:Extensions}}

\begin{lem}
  In the non-preemptive setting, at~$1+O(\eps)$ loss we can ensure
  that at the end of each interval~$I_{x}$,
  \begin{compactitem}
  \item there are at most~$m$ large jobs from each type which are
    partially (i.e., neither fully nor not at all) processed, and
  \item for each partially but not completely processed large job~$j$
    there is a value~$k_{x,j}$ such that~$j$ is processed for at
    least~$k_{x,j}\cdot p_{j}\cdot\N$ time units in~$I_{x}$,
  \item we calculate the objective with adjusted completion
    times~$\bar{C}_{j}=R_{c(j)}$ for some value~$c(j)$ for each
    job~$j$ such that~$\sum_{x<c(j)}k_{x,j}\cdot p_{j}\cdot\N\ge
    p_{j}$.
  \end{compactitem}
\end{lem}
\begin{proof}
  Note that the first property holds for any non-preemptive schedule
  and is listed here only for the sake of clarity. The other two
  properties can be shown similiarly as in the proof of
  Lemma~\ref{lem:large-job-atoms}.
\end{proof}

\begin{proof}[Proof of Lemma~\ref{lem:split-the-instance-non-pmtn}]
  Assume that we have an online algorithm~$\alg$ with competitive
  factor~$\rho_{\alg}$ on instances in which for every~$i$ the first
  job~$\first{i}$ released in part~$\Part{i}$
  satisfies~$\sum_{\ell=1}^{i-1} rw(\Part{\ell}) \le w_{\!\first{i}}
  \epsilon/\ee^{s}$~(i.e.,~$\first{i}$ dominates all previously
  released parts).  Based on $\alg$ we construct a new
  algorithm~$\alg'$ for arbitrary instances with competitive ratio at
  most~$\ee \rho_{\alg}$: When a new part~$\Part{i}$ begins, we scale
  the weights of all jobs in~$\Part{i}$ such
  that~$\sum_{\ell=1}^{i-1}rw(\Part{\ell}) \le w'_{\!\first{i}}
  \epsilon/\ee^{s}$, where the values~$w'_{\!j}$ denote the adjusted
  weights. Denote by~$\bar{I}(i)$ the resulting instance up to~(and
  including) part~$\Part{i}$. We schedule the resulting instance
  using~$\alg$.  We take the computed schedule for each part
  $\Part{i}$ and use it for the jobs with their original weight,
  obtaining a new algorithm~$\alg'$. The following calculations shows
  that this procedure costs only a factor $\e$. To this end, we proof
  that for any instance $I$ it holds that
$$
\frac{\alg'(\I)}{\opt(\I)} \le
\max_{i}\frac{\alg(\bar{I}(i))}{\opt(\bar{I}(i))}\cdot(1+O(\epsilon))
\le (1+O(\epsilon))\rho_{\alg}.
$$
For each~$\Part{i}$ we define~$\alg'(\I|\Part{i})$ to be the amount that the
jobs in~$\Part{i}$ contribute in~$\alg'(\I)$. Similarly, we
define~$\opt(\I|\Part{i})$. We have that
  \[
  \frac{\alg'(\I)}{\opt(\I)} \le
  \max_{i}\frac{\alg'(\I|\Part{i})}{\opt(\I|\Part{i})} \le
  \max_{i}\frac{\alg'(\I|\Part{i})}{\opt(\Part{i})}.
  \]
  We claim that for each~$i$ holds~$\frac{\alg'(\I|\Part{i})}{\opt(\Part{i})}
  \le (1+O(\epsilon))\cdot \frac{\alg(\bar{I}(i))}{\opt(\bar{I}(i))}$.
  For each part~$\Part{i}$ let~$v_{i}$ denote the scale factor of the
  weight of each job in~$\bar{I}(i)$ in comparison to its original
  weight.  The optimum for the instance~$\bar{I}(i)$ can be bounded by
  \begin{eqnarray*}
    \opt(\bar{I}(i)) & \le & \opt(\Part{i})\cdot v_{i} +
    \ee^{s}\sum_{\ell=1}^{i-1}rw(\Part{\ell}) \leq
    \opt(\Part{i})\cdot v_{i}+\epsilon\cdot r_{\first{i}} \cdot w'_{\first{i}}
    \le \ee \opt(\Part{i})\cdot v_{i}.
  \end{eqnarray*}
  Furthermore holds by construction~$\alg'(\I|\Part{i})\cdot
  v_{i}\le\alg(\bar{I}^{i})$.
  Thus,~$\max_{i}\frac{\alg'(I|\Part{i})}{\opt(\Part{i})} \le
  \max_{i}\frac{\alg(\bar{I}^{i})}{\opt(\bar{I}^{i})}\cdot(1+O(\epsilon))$.
\end{proof}

\begin{proof}[Proof of Lemma~\ref{lem:Qm-pmtn-speed-bound}]
  Given a schedule on related machines with speed
  values~$s_{1},...,s_{\max}$, we stretch time twice. Thus, we
  gain in each interval~$I_{x}$ free space of size~$\varepsilon I_{x}$
  on the fastest machine.  For each machine whose speed is at
  most~$\frac{\epsilon}{m}s_{\max}$,
we take its schedule of the
  interval~$I_{x}$ and simulate it on the fastest machine.
  Thus, those slow machines are not needed and can be removed. The
  remaining machines have speeds in~$[\frac{\epsilon}{m}\,
  s_{\max},s_{\max}]$. Assuming the slowest machines has unit speed
  gives the desired bound.
\end{proof}

\begin{proof}[Proof of Lemma~\ref{lem:Rm-range-processing-times}]
  Consider a schedule for an instance which does not satisfy the
  property.  We stretch time twice and thus we gain a free
  space of~$\epsilon I_{x}$ in each interval~$I_{x}$. Consider
  some~$I_{x}$ and a job~$j$ which is scheduled in~$I_{x}$.  Let~$i$
  be a fastest machine for~$j$. We remove the processing volume of~$j$
  scheduled in~$I_x$ on slow machines~$i'$ with~$p_{i'j} >
  \frac{m}{\eps}\, p_{ij}$ and schedule it on~$i$ in the gained free
  space. This way, we obtain a feasible schedule even if a job never
  runs on a machine where it is slow. Thus, we can
  set~$p_{i'j}=\infty$ if there is a fast machine~$i$ such that
  $p_{ij} \le \frac{\epsilon}{m}p_{i'j}$.
\end{proof}

\section{Proofs of Section~\ref{sec:randomized}}

\begin{proof}[Proof of Lemma~\ref{lem:rand-const-periods}]
  Let~$\alg$ be a randomized online algorithm with a competitive ratio
  of~$\rho_{\alg}$ on instances which span at most~$\ee^{s}/\eps$
  periods. We construct a new randomized online algorithm~$\alg'$ which
  works on arbitrary instances~$\I$ such
  that~$\rho_{\alg'}\le\rho_{\alg}(1+\eps)$.  At the beginning of~$\alg'$,
  we choose an offset~$o\in\{0,...,M-1\}$ uniformly at random
  with~$M:=\left\lceil \ee^{s}/\eps\right\rceil$.  In
  instance~$\I$, we move all jobs to their safety net which are released in
  periods~$Q\in\mathcal{Q}:=\{Q_{i}|i\equiv o\bmod M\}$. This splits
  the instance into parts~$\Part{0},...,\Part{k}$ where each part~$\Part{\ell}$
  consists of the periods~$Q_{o+(\ell-1)\cdot M},...,Q_{o+\ell\cdot
    M-1}$.  Note that at the end of each part no job remains. We need
  to bound the increase in the total expected cost caused by moving
  all jobs in periods in~$\mathcal{Q}$ to their safety nets. This
  increase is bounded by

 \begin{eqnarray*}
   \mathbb{E}\left[\sum_{Q\in\mathcal{Q}}\sum_{j\in Q}\ee^{s}r_{j}\cdot
w_{j}\right] & \le &
\ee^{s}\mathbb{E}\left[\sum_{Q\in\mathcal{Q}}rw(Q)\right]\\
   & \le & \ee^{s}\frac{1}{M}\sum_{Q\in \I}rw(Q)\\
   & \le & \eps\cdot rw(\I)\\
   & \le & \eps\cdot \opt(\I)\,.
 \end{eqnarray*}

 Thus, the total expected cost of the computed schedule is
 \begin{eqnarray*}
   \mathbb{E}\left[\eps\cdot \opt(\I)+\sum_{i=1}^{k}\alg(\Part{i})\right] & \le
& \eps\cdot \opt(\I)+\sum_{i=1}^{k}\rho_{\alg}\cdot \opt(\Part{i})\\
   & \le & \eps\cdot \opt(\I)+\rho_{\alg}\cdot \opt(\I)\\
   & \le & (\rho_{\alg}+\eps)\cdot \opt(\I)\\
   & \le & \rho_{\alg}(1+\eps)\cdot \opt(\I).
 \end{eqnarray*}
 Thus, at~$1+\eps$ loss in the competitive ratio we can restrict to parts~$I_i$ which span a constant number of periods.
\end{proof}

\begin{proof}[Proof of Lemma~\ref{lem:discretize-prob}]
  Consider an instance~$\I$. Let~$\delta>0$ and~$k\in\mathbb{N}$ be values to be determined later
  with the property that~$1/\delta\in\mathbb{N}$. For each configuration~$C$ and each
  interval-schedule $S$ we
  define a value~$g(C,S)$ such that~$\left\lfloor
    \frac{f(C,S)}{\delta}\right\rfloor \cdot\delta\le
  g(C,S)\le\left\lceil \frac{f(C,S)}{\delta}\right\rceil \cdot\delta$
  and~$\sum_{S\in\mathcal{S}}g(C,S)=1$. Now we want to
  bound~$\rho_{g}$.

  The idea is that for determining the
  ratio~$\Exp{g(\I)}/\opt(\I)$ %
  it suffices to consider schedules~$S(\I)$ which are computed with
  sufficiently large probability. We show that also~$f$ computes them
  with almost the same probability. Let~$S(\I)$ denote a schedule for
  the entire instance~$\I$. We denote by~$P_{f}(S(\I))$
  and~$P_{g}(S(\I))$ the probability that~$f$ and~$g$ compute the
  schedule~$S(\I)$ when given the instance~$\I$. Assume
  that~$P_{f}(S(\I))\ge k\cdot\delta$.  Denote by~$C_{0},...,C_{\bar{\Gamma}-1}$
  the configurations that algorithms are faced with when
  computing~$S(\I)$, i.e., each configuration~$C_{x}$ contains the
  jobs which are released but unfinished at the beginning of
  interval~$I_{x}$ in~$S(\I)$ and as history the schedule~$S(\I)$
  restricted to the intervals~$I_{0},...,I_{x-1}$. Denote by~$S_{\!x}$
  the schedule of~$S(\I)$ in the interval~$I_{x}$.
  Hence,~$P_{f}(S(\I))=\prod_{x=0}^{\bar{\Gamma}-1}f(C_{x},S_{\!x})$.  Note that
  from $P_{f}(S(\I))\ge k\cdot\delta$ follows that~$f(C_{x},S_{\!x})\ge
  k\cdot\delta$ for all~$x$. For these schedules,~$P_{g}(S(\I))$ is
  not much larger since
  \begin{eqnarray*}
    P_{g}(S(\I)) & = & \prod_{x=0}^{\bar{\Gamma}-1}g(C_{x},S_{\!x})
    \le  \prod_{x=0}^{\bar{\Gamma}-1}\frac{k+1}{k}f(C_{x},S_{\!x})
    \le  \left(\frac{k+1}{k}\right)^{\bar{\Gamma}}P_{f}(S(\I)).
  \end{eqnarray*}
  Let~$\mathcal{S}(\I)$ denote the set of all schedules for~$\I$. We
  partition~$\mathcal{S}(\I)$ into
  schedule sets~$\mathcal{S}_{H}^{g}(\I):=\{S(\I)|P_{g}(S(\I))\ge
  k\cdot\delta\}$
  and~$\mathcal{S}_{L}^{g}(\I):=\mathcal{S}(\I)\setminus\mathcal{S}_{H}(\I)$.
    We estimate the expected value of a schedule computed by algorithm map~$g$ on~$\I$ by

  \begin{eqnarray*}
    \Exp{g(\I)} & = & \sum_{S(\I)\in\mathcal{S}_{H}^{g}(\I)}P_{g}(S(\I))\cdot S(\I) + \sum_{S(\I)\in\mathcal{S}_{L}^{g}(\I)}P_{g}(S(\I)) \cdot S(\I)\\
    & \le &
\sum_{S(\I)\in\mathcal{S}_{H}^{g}(\I)}P_{f}(S(\I))\cdot\left(\frac{k+1}{k}
\right)^{\bar{\Gamma}} \cdot S(\I) + |\mathcal{S}(\I)| \cdot k\cdot\delta\cdot\ee^{s} \cdot rw(\I)\\
    & \le &
\left(\frac{k+1}{k}\right)^{\bar{\Gamma}}\sum_{S(\I)\in\mathcal{S}_{H}^{g}(\I)}P_{f}
(S(\I))\cdot S(\I) + |\mathcal{S}(\I)|\cdot k\cdot\delta\cdot\ee^{s}\cdot rw(\I)\\
    & \le & \left(\frac{k+1}{k}\right)^{\bar{\Gamma}} \Exp{f(\I)} + |\mathcal{S}(\I)|\cdot
k\cdot\delta\cdot\ee^{s}\cdot rw(\I).
  \end{eqnarray*}
  We choose~$k$ such that~$\left(\frac{k+1}{k}\right)^{\bar{\Gamma}}\le1+\eps/2$
  and~$\delta$ such that~$|\mathcal{S}(\I)|\cdot
  k\cdot\delta\cdot\ee^{s}\le\eps/2$ for all instances~$\I$~(note
  here that~$|\mathcal{S}(\I)|$ can be upper bounded by a value
  independent of~$\I$ since our instances contain only constantly many
  jobs). This yields
    \[
  \frac{\Exp{g(\I)}}{\opt(\I)} \le
  (1+\eps/2)\cdot \frac{\Exp{f(\I)}}{\opt(\I)} + \eps/2\cdot\frac{rw(\I)}{\opt(\I)}    \le(1+\eps) \cdot \frac{\Exp{f(\I)}}{\opt(\I)}\,,
  \]
  and we conclude that~$\rho_{g}\le\ee\rho_{f}$.
\end{proof}

\section{Competitive-Ratio Approximation Schemes for Minimizing $C_{\max}$ (cf. Section \ref{sec:other-objectives})}
Consider the objective of minimizing the makespan. The {\em simplifications within intervals} of Section~\ref{sec:simplifications} are based on arguing on completion times of individual jobs, and clearly hold also for the last job. Thus, they directly apply to minimizing the makespan. 

Furthermore, we simplify the definition of {\em irrelevant history} in Section~\ref{sec:simplifications} by omitting the partition of the instance into parts. %
We observe that when then last job is released at time $R_{x^{*}}$ then all jobs $j$ with
$r_{j}\le R_{x^{*}-s}$ are irrelevant for the objective: such a job
$j$ finishes at time $R_{x^{*}}$ the latest in any schedule (due
to the safety net) and $OPT\ge R_{x^{*}}$. Therefore, we define a
job $j$ to be \emph{irrelevant at time $R_{x}$ }if $r_{j}\le R_{x-s}$.
We keep Definition~\ref{def:equivalence-interval-schedule} for the equivalence relation of schedules as it is except for the notion of job weights which are not important for the makespan.
Based on the above definition for relevant jobs we define equivalence
classes of configurations. With this definition, we can still restrict
to algorithm maps $f$ with $f(C)\sim f(C')$ for any two equivalent
configurations $C,C'$ (Lemma~\ref{lem:equal-confs}). Lemmas \ref{lem:const-numb-maps} to \ref{lem:approx-comp-factor} then hold accordingly.
Finally, note that since we do not split the instance into parts,
we do not need (an adjusted version of) Lemma~\ref{lem:split-the-instance-non-pmtn} in the non-preemptive
case.
\begin{thm}
For any $m\in \mathbb{N}$ we obtain competitive-ratio approximation schemes for~\abc{Pm}{r_{j},(pmtn)}{C_{\max}}, \abc{Qm}{r_{j},pmtn}{C_{\max}}, and \abc{Rm}{r_{j},pmtn}{C_{\max}}, and for \abc{Qm}{r_{j}}{C_{\max}} assuming a constant range of machine speeds.
\end{thm}
For constructing randomized online algorithm schemes for minimizing
the makespan, similarly to Lemma~\ref{lem:rand-const-periods} we can show
that we can restrict our attention to instances which span only a constant
number of periods.

\begin{lem}\label{lem:rand-const-periods-makespan} For randomized
algorithms for minimizing the makespan, at~$\eo$ loss we can restrict to instances in which all
jobs are released in at most~$\ee^{s}/\eps$ consecutive periods.
\end{lem}
\begin{proof}
Let~$\alg$ be a randomized online algorithm for minimizing the makespan
over time with a competitive ratio of~$\rho_{\alg}$ on instances
which span at most~$\ee^{s}/\eps$ periods. We construct a new randomized
online algorithm~$\alg'$ which works on arbitrary instances
such that~$\rho_{\alg'}\le\rho_{\alg}(\eo)$. Our reasoning 
is similar to the proof of Lemma~\ref{lem:rand-const-periods}.

At the beginning of~$\alg'$, we choose an offset~$o\in\{0,...,M-1\}$
uniformly at random with~$M:=\left\lceil \ee^{s}/\eps\right\rceil $.
Given an instance~$\I$, we move all jobs to their safety net which are
released in periods~$Q\in\mathcal{Q}:=\{Q_{i}|i\equiv o\bmod M\}$.
This splits the instance into parts~$\Part{0},...,\Part{k}$ where
each part~$\Part{\ell}$ consists of the periods~$Q_{o+(\ell-1)\cdot M},...,Q_{o+\ell\cdot M-1}$.
Note that at the end of each part no job remains. We present each
part separately to $\alg$.

We bound the competitive ratio $\rho_{\alg'}$ of the resulting algorithm.
Let $R_{x^{*}}:=\max_{j \in \I}r_{j}$. Moving the jobs from periods $\mathcal{Q}$
has an effect on the optimal makespan only if $o$ is chosen such
that at least one job $j$ with $r_{j}>R_{x^{*}-s}$ is moved. 
There are at most two offsets $o$ such that this
happens. In
that case, the algorithm still achieves a competitive ratio of at
most $\ee^{s}$. In all other cases, $\alg'$ achieves a competitive ratio
of at most~$\rho_{\alg}$. Thus, we can bound $\rho_{\alg'}$ by
\begin{eqnarray*}
\rho_{\alg'} & \le & \frac{2}{M}\ee^{s}+\frac{M-2}{M}\rho_{\alg}\le2\eps+\rho_{\alg}\le (\eo)\rho_{\alg}.
\end{eqnarray*}
\end{proof}

We can prove similarly as in Lemma~\ref{lem:discretize-prob} that any
randomized algorithm map $f$ can be well approximated by a discretized
randomized algorithm map~$g$. Hence, we obtain the following theorem.
\begin{thm}
For any $m\in \mathbb{N}$ we obtain randomized competitive-ratio approximation schemes for \abc{Pm}{r_{j},(pmtn)}{C_{\max}},
\abc{Qm}{r_{j},pmtn}{C_{\max}}, and \abc{Rm}{r_{j},pmtn}{C_{\max}}, and for \abc{Qm}{r_{j}}{C_{\max}} assuming a constant range of machine speeds.
\end{thm}

\end{document}